\documentclass[sigconf, nonacm]{acmart}

\usepackage{graphicx}

\usepackage{balance}  

\usepackage{xcolor}
\usepackage{subcaption}
\usepackage{times}

\usepackage{latexsym}   
\usepackage{amsmath}    
\usepackage{fancyvrb}   
\usepackage{float}      

\usepackage{algorithm}
\usepackage{algorithmic}
\usepackage{caption}
\usepackage{subcaption}
\usepackage{paralist}
\usepackage{enumerate}
\usepackage{geometry}
\usepackage{color}
\usepackage{graphicx}
\usepackage{enumitem}

\usepackage{multirow}

\usepackage{amsthm}
\newtheorem{definition}{Definition}
\newtheorem{example}{Example}
\newtheorem{lemma}{Lemma}
\newtheorem{theorem}{Theorem}

\usepackage{sty/term}       

\usepackage{sty/local}      

\newcommand{\blue}[1]{\textcolor{black}{#1}}
\newcommand{\diagram}[1]{\textcolor{black}{#1}}

\newenvironment{change}{\par\color{black}}{\par}
\newcommand{\paddingT}{\vskip -0.125in}
\newcommand{\paddingD}{\vskip -0.13in}

\begin{document}


\title{Discovering Domain Orders through Order Dependencies}

\author{Reza Karegar}
\affiliation{%
  \institution{University of Waterloo, CA}
}
\email{mkaregar@uwaterloo.ca}
\author{Melicaalsadat Mirsafian}
\affiliation{%
  \institution{University of Waterloo, CA}
}
\email{mmirsafian@uwaterloo.ca}
\author{Parke Godfrey}
\affiliation{%
  \institution{York University, CA}
}
\email{godfrey@yorku.ca}
\author{Lukasz Golab}
\affiliation{%
  \institution{University of Waterloo, CA}
}
\email{lgolab@uwaterloo.ca}
\author{Mehdi Kargar}
\affiliation{%
  \institution{Ryerson University, CA}
}
\email{kargar@ryerson.ca}
\author{Divesh Srivastava}
\affiliation{%
  \institution{AT\&T Labs-Research, US}
}
\email{divesh@research.att.com}
\author{Jaroslaw Szlichta}
\affiliation{%
  \institution{Ontario Tech Univ, CA}
}
\email{jarek@ontariotechu.ca}

 \author{
 }

\begin{abstract}
Much real-world data come with explicitly defined domain orders; e.g., lexicographic order for strings, numeric for integers, and chronological for time. Our goal is to discover implicit domain orders that we do not already know; for instance, that the order of months in the 
Chinese 
Lunar calendar is \emph{Corner} $\prec$ \emph{Apricot} $\prec$ \emph{Peach}.
To do so, we enhance data profiling methods by discovering implicit domain orders in data through order dependencies.
We enumerate tractable special cases and 
proceed towards the most general case, which we prove is NP-complete. 
We show that the general case
nevertheless 
can be effectively handled by a SAT solver. 
We also
devise
an interestingness measure to rank the discovered implicit domain orders,  
\blue{which we validate with a user study}.
\blue{Based on an extensive suite of experiments with real-world data, we establish the efficacy of our algorithms, and the utility of the domain orders discovered by demonstrating significant added value in three applications (data profiling, query optimization, and data mining).}

\end{abstract}

\maketitle


\section{Introduction} \label{SEC/intro}

Much real-world data come with explicitly defined orders; e.g., lexicographic for strings, numeric for integers and floats, and chronological for time. 
Our goal is to go a step further to discover potential domain orders that are not already known. We call these \emph{implicit orders}.  
Consider Table~\ref{table:festivals}, describing festivals in various countries.  The $\A{timestamp}$ column has an explicit chronological order.  Given this explicit order, we show how to discover the \emph{implicit} order of months in the Gregorian calendar, $\A{monthGreg}$ (\emph{January} $\prec$ \emph{February} $\prec$ \emph{March}, etc).  Moreover, we will show how to find implicit orders based on other implicit orders.  For instance, given the implicit order on $\A{monthGreg}$, we can find the implicit order of months in the traditional Chinese (Lunar) calendar, $\A{monthLun}$ 
	(\(
		\mbox{\emph{Corner}}
			\prec
		\mbox{\emph{Apricot}}
			\prec
		\mbox{\emph{Peach}}
			\prec
		\mbox{\emph{Plum}}
			\prec
		\mbox{\emph{Pomegranate}}
			\prec
		\mbox{\emph{Lotus}}
			\prec
		\mbox{\emph{Orchid}}
			\prec
		\mbox{\emph{Osmanthus}}
			\prec
		\mbox{\emph{Chrysanthemum}}
			\prec
		\mbox{\emph{Dew}}
			\prec
		\mbox{\emph{Winter}}
			\prec
		\mbox{\emph{Ice}}
	\)).

\blue{Domain orders are useful in a number of important applications:} 
\begin{itemize}[nolistsep,leftmargin=*]

\item Implicit orders can enhance data profiling methods to identify new data quality rules, such as order dependencies over implicitly ordered attributes. \blue{(See Section~\ref{sec:exp:apps}, Exp-\ref{exp-data-profiling} and Exp-\ref{exp-yago}.)}

\item The \SQL\ standard includes an order-by clause
to sort the output, and aggregation with respect to minimum and maximum values requires a domain order. 
By capturing relationships between ordered attributes, we can eliminate the 
necessity to sort
if the query plan already produces results in 
a needed order%
\blue{. These applications have been investigated for explicit orders in \cite{QO-exp-orders-vldb05,SGG+2013:complexity,QO-edbt11}.
In this work, we demonstrate the additional benefits of implicit orders.}
\blue{(See Section~\ref{sec:exp:apps}, Exp-\ref{exp-query-optimization}.)}



\item \blue{Implicit orders can improve the performance of machine learning techniques by turning categorical features into ordinal ones. One case of this is the splitting condition in decision trees. Similarly, implicit orders can produce concise data summaries,  with ranges over ordered attributes instead of individual categories.  
We demonstrate this by enhancing a recent data summarization method \cite{info-summarization-ordinal} with newly discovered implicit orders.
(See Section~\ref{sec:exp:apps}, Exp-\ref{exp-data-mining}.)
}



\end{itemize}

\bgroup
\tabcolsep=0.95\tabcolsep
\begin{table*}[t]
\begin{small}
\paddingT
\caption{A dataset with festival information in various countries.}
\paddingT
\label{table:festivals}
\scalebox{0.87}{
    \begin{tabular}{rl|l|l|l|l|l|l|l|l|l|l|l|l}
                 \multicolumn{1}{c}{\relax}         &
                 \multicolumn{1}{c}{\A{festival}}   &
                 \multicolumn{1}{c}{\A{country}}    &
                 \multicolumn{1}{c}{\A{timestamp}}  &
                 \multicolumn{1}{c}{\A{week}}       &
                 \multicolumn{1}{c}{\A{yearGreg}}       &
                 \multicolumn{1}{c}{\A{monthGreg}}  &
                 \multicolumn{1}{c}{\A{yearLun}} &
                 \multicolumn{1}{c}{\A{monthLun}} &
                 \multicolumn{1}{c}{\A{count}}      &
                 \multicolumn{1}{c}{\A{size}}       &
                 \multicolumn{1}{c}{\A{ribbon}}     &
                 \multicolumn{1}{c}{\A{profit}}     &
                 \multicolumn{1}{c}{\A{tax}}        \\
        
        \cline{2-14}
        
        $\tup{t}[1]$ &\em New Year Eve Shanghai &  China &  20200125 & 4 & 2020 & January & 4718 & Corner & 10M & X-Large & Blue & 30M & 2.7M \Tstrut\\

        $\tup{t}[2]$ &\em Tomb Sweeping Day Xi & China &  20200404 & 14 & 2020 & April & 4718 & Peach & 650K & Medium & Red & 1.3M & 117K \Tstrut\\
        
         $\tup{t}[3]$ &\em Buddha B~day Liuzhou &  China &  20200430 & 18 & 2020 & April & 4718 & Plum  &  450K & Medium & White & 900K & 81K \Tstrut\\
        

        $\tup{t}[4]$ &\em Dragon Boat Hangzhou & China &  20200625 & 26 & 2020 & June & 4718 & Pomegranate &  2M & X-Large & Red & 3M & 270K \Tstrut\\
        
        $\tup{t}[5]$ &\em Dongzhi Festival &  China &  20201221 & 52 & 2020 & December & 4718 & Winter &  950K & Large & Red & 3.5M & 315K \Tstrut\\

        \cline{2-14}

        $\tup{t}[6]$ &\em New Year Quebec &  Canada &  20210101 & 1 & 2021 & January & 4718 & Winter & 800K & Large & Red & 3M & 390K \Tstrut\\
        
        $\tup{t}[7]$ &\em TD Toronto Jazz &  Canada & 20210618 & 25 & 2021 & June & 4719 & Pomegranate &  500K & Medium & Blue & 1.5M & 195K \Tstrut\\

        $\tup{t}[8]$ &\em Rogers Tennis Cup &  Canada &  20210807 & 32 & 2021 & August & 4719 & Lotus  &  600K & Medium &  Red & 1.2M & 156K \Tstrut\\

        $\tup{t}[9]$ &\em Steam Era Ontario &  Canada &  20210830 & 36 & 2021 & August & 4719 & Osmanthus &  125K & Small & Blue & 1M & 130K \Tstrut\\

        $\tup{t}[10]$ &\em Octoberfest Waterloo &  Canada &  20211009 & 41 & 2021 & October & 4719 & Chrysanthemum & 50K & Small & White & 150K & 19.5K \Tstrut\\
        





    \end{tabular}
}
\end{small}
 \centering
\end{table*}
\egroup

\subsection{Methodology} \label{SEC/intro/methodology}

Manual 
specification 
does not scale 
\blue{as it requires 
domain experts \cite{TANE,Go2009,LN16,SGG+2017:ODD}}.
This motivates the need to discover implicit orders automatically.
\blue{\emph{Integrity constraints} (ICs) specify relationships between attributes in databases \cite{IC-def}}.
\blue{To discover implicit orders, we use}
\emph{order dependencies}  ($\OD$s) which capture relationships between orders \blue{\cite{od-def1-pointwise,od-def2-lexicographic}}.
%
%
\blue{Intuitively, an order dependency asserts that sorting a table according to some attribute(s) implies that the table is also sorted according to some other attribute(s).}
\blue{For instance,} 
in Table~\ref{table:festivals}, $\A{timestamp}$ orders $\A{yearGreg}$.
%
If the $\A{tax}$ per $\A{country}$ is a fixed or a progressive percentage of the $\A{profit}$, then sorting the table by $\A{country}$, $\A{profit}$ results in the table also being sorted by $\A{country}$, $\A{tax}$. Hence, ``$\A{country}$, $\A{profit}$ \emph{orders} $\A{country}$, $\A{tax}$.'' 
The order of attributes on the left- and right-hand sides in an $\OD$ matters, as in the SQL \emph{order-by} clause, while the order of attributes in a functional dependency ($\FD$) \blue{\cite{fd-def}} does not, as in the SQL \emph{group-by} clause.
%

An $\OD$ implies the corresponding $\FD$, modulo lists and sets of attributes but, not vice versa; e.g., $\A{country}$, $\A{profit}$ orders $\A{country}$, $\A{tax}$ implies that $\A{country}$, $\A{profit}$ functionally determines $\A{country}$, $\A{tax}$. 
\emph{Order compatibility} ($\OC$) \cite{Szl2012}
is a weaker version of
an $\OD$, without the implied $\FD$. 
\blue{
Two lists of attributes in a table are said to be \emph{order compatible} if there exists an arrangement for the tuples in the database in which the tuples are sorted according to both of the lists of attributes.
}
\blue{
For instance,}
$\A{yearGreg}$, $\A{monthNum}$ \emph{is order compatible} with $\A{yearGreg}$, $\A{week}$, where the attribute
$\A{monthNum}$ (not included in Table~\ref{table:festivals}) denotes the Gregorian month of the year in numeric format and $\A{week}$ represents the week of the year.
A corresponding $\FD$ does \emph{not} hold: $\A{yearGreg}$, $\A{monthNum}$ does not functionally determine $\A{yearGreg}$, $\A{week}$ (there are multiple weeks in a month) and $\A{yearGreg}$, $\A{week}$ does \emph{not} functionally determine $\A{yearGreg}$, $\A{monthNum}$ (a week may span two months).

\blue{
When an $\OD$ or $\OC$ has a common prefix on its left- and right-side, 
we can ``factor out'' the common prefix to increase understandability and 
refer to it as the \emph{context}. 
Intuitively, this means that the respective $\OD$ or $\OC$ holds \emph{within} each partition group of data by the context.
}
\blue{
For instance, if $\A{country}$, $\A{profit}$ orders $\A{country}$, $\A{tax}$, 
then given a partitioning of the data by $\A{country}$ (i.e., the context), 
$\A{profit}$ orders $\A{tax}$
within each \emph{group}
(that is, for any given \emph{country}).
}
When an \OD\ or \OC\ has no common prefix,
we say it has an \emph{empty} context;
e.g., the \OD\ of \A{timestamp}\ orders \A{yearGreg}\
has no common prefix, and thus an empty context.


Algorithms for $\OD$ and $\OC$ discovery from data 
\cite{%
	LN16,%
	SGG+2017:ODD,%
	SGG2018:BiODD,%
	pointwise-od-discovery-vldb20
}
use \emph{explicit} domain orders. 
Let us say that they
discover \emph{explicit-to-explicit} ($\A{E/E}$) \OD s.
We discover implicit orders by extending the machinery of $\OD$ discovery.
%
We first leverage explicitly known domain orders,
where, say,
the left-hand-side of a ``candidate'' $\OC$ is an explicit domain order
and the right-hand side is a learned, implicit domain order.
Call this an \EIOC. 
For instance, in the context of $\A{yearGreg}$, $\A{timestamp}$ is order compatible with $\A{monthGreg}^{*}$, where the star denotes implicit domain order over an attribute.
%
Astonishingly,
implicit domain orders
can be also discovered from a ``candidate'' \OC\
for \emph{both} the left- and right-hand sides of the \OC!
Call this an \IIOC.
%
	\EIOD s and \IIOD s
	simply are \EIOC s and \IIOC s, respectively,
	for which
	there is also an \FD\ from the lefthand side to the right.
For example, 
in the context of $\A{yearGreg}$ and $\A{yearLun}$, $\A{monthGreg}^{*}$  is order compatible with $\A{monthLun}^{*}$. 
%


\subsection{\blue{Overview} and Contributions}
\label{SEC/intro/contributions}

Our goal is to \textbf{discover implicit domain orders}.
To do this,
we define candidate classes for \EIOC s and \IIOC s, and we extend the discovery methods for these.
To the best of our knowledge, this is the first attempt to discover implicit domain orders through ICs.
The problem space can be factored by the following dimensions.
\begin{itemize}[nolistsep,leftmargin=*]
\item
	Whether there is a corresponding \FD\
	(thus, finding \OD s rather than \OC s).

\item
	Whether the context is \emph{empty}.

\item
    \blue{When} the context is non-empty,
	whether the discovered domain orders
	across \blue{different} partition groups
	with respect to the context
	are \blue{to be} considered as independent \blue{of each other},
	and so can be different
	(\emph{conditional}),
	or they are \blue{to be} considered \blue{the same across partition groups}, 
	and must be consistent
	(\emph{unconditional}). In Table~\ref{table:festivals}, the implicit order $\A{monthGreg}^{*}$ is unconditional; however, the implicit order $\A{ribbon}^{*}$ is conditional in the relative context of the country with respect to the size of the festival, with \emph{White} $\prec$ \emph{Blue} $\prec$ \emph{Red} in Canada and \emph{White} $\prec$ \emph{Red} $\prec$ \emph{Blue} in China.%
	\footnote{%
    	This example is inspired by equine competitions 
    	(\url{https://en.wikipedia.org/wiki/Horse_show\#Awards}).%
	}

\item
	Whether we are considering \EIOC s
	or \IIOC s.

\end{itemize}




Our key contributions 
are as follows.

\begin{enumerate}[start,nolistsep,leftmargin=*]
\item 
    \textbf{\blue{Implicit} Domain Orders} 
	(\emph{Section \ref{sec:prelims}}).
	
	\noindent
	We formulate a novel data profiling problem of discovering implicit domain orders through a significant broadening of \OD\ / \OC\ discovery, 
	and we parameterize the problem space. 
	
\end{enumerate}

\noindent We divide the problem space
between \emph{explicit-to-implicit}
and
\emph{implicit-to-implicit}, which we present in Sections \ref{sec:eioc} and \ref{sec:iioc}, respectively.
We identify tractable cases, and then proceed towards the general case of \IIOC\ discovery,
which we prove is NP-complete.

\begin{enumerate}[resume,nolistsep,leftmargin=*]
\item 
    \textbf{E/I Discovery}
    (\emph{Section \ref{sec:eioc}}).
    
    \noindent
	For implicit domain order discovery through \EIOD s and \EIOC s,
	we present efficient algorithms, taking polynomial time in the number of tuples to verify a given $\OD$ or $\OC$ candidate. 

\item 
	\textbf{I/I Discovery} (\emph{Section \ref{sec:iioc}}).
	
	\noindent
	For implicit domain order discovery through \IIOC s,

	\begin{enumerate}[nolistsep,leftmargin=*]
	\item 
		we present a polynomial candidate verification algorithm
		when the context is empty,

	\item 
		we present
		a polynomial candidate verification algorithm
		when the context is non-empty
		but taken as \emph{conditional},

	\item 
		we prove that, 
		for non-empty contexts taken as \emph{unconditional},
		the problem is NP-complete,
		\emph{and}

	\item 
		we show why the candidate set of conditional \IIOD s is always empty, although it is not necessarily empty for unconditional \IIOD s.

	\end{enumerate}

\item 
    \textbf{Algorithmic Approaches}
    (\emph{Sections \ref{SEC:SAT-Solver} \& \ref{sec:measure-interestingness-def}}).
    
    \begin{enumerate}[nolistsep,leftmargin=*]
    \item While the general case of implicit order discovery through \IIOC s is NP-hard, 
	we show that the problem can be effectively handled by a SAT solver (Section~\ref{SEC:SAT-Solver}). 
	We implement our methods in a lattice-based framework that has been used to mine $\FD$s and $\OD$s from data \cite{SGG2018:BiODD,LN16,SGG+2017:ODD}.
	\item We propose an \emph{interestingness measure} to rank the discovered orders and simplify manual validation  (Section~\ref{sec:measure-interestingness-def}).
    \end{enumerate}
    
	%
    %
    
	%

\item 
    \textbf{Experiments}
    (\emph{Section \ref{sec:experiments}}).
    
    \noindent
    We mirror the sub-classes and approaches to discover implicit domain orders defined
    in Sections \ref{sec:eioc} and \ref{sec:iioc}
    and the algorithms
    in Sections \ref{SEC:SAT-Solver} and \ref{sec:measure-interestingness-def} 
    via experiments over real-world datasets.
	\begin{enumerate}[nolistsep,leftmargin=*]
	\item \blue{\textbf{Scalability:} In Section~\ref{sec:exp:scal}, we demonstrate  scalability in the number of tuples and attributes, and the effectiveness of our method for handling  NP-complete instances.
	}
	\item \blue{\textbf{Effectiveness:} In Section~\ref{sec:exp:lattice}, we validate the utility of the discovered orders.
	}
	\item \blue{\textbf{Applications:} In Section~\ref{sec:exp:apps}, we demonstrate the usefulness of implicit orders in data profiling (by finding more than double the number of data quality rules by involving orders not found in existing knowledge bases), query optimization (by reducing query runtime by up to 30\%), and data summarization (by increasing the information contained in the summaries by an average of 60\%).}

	\end{enumerate}
	
\end{enumerate}

\noindent
We review related work in Section \ref{sec:relWork} 
and conclude
in Section \ref{sec:futureWork}.

\section{Preliminaries} \label{sec:prelims}

\subsection{Domain Orders and Partitions}

\blue{We first review definitions of order, as introduced in \cite{order-book}.}
\blue{A glossary of relevant notation is provided in Table~\ref{tab:notations}.}
If we can write down a sequence of a domain's values
to represent how they are ordered%
, 
this defines a \emph{strong} total order over the values.
For two distinct \emph{dates},
for example,
one always precedes the other in time.
%

\begin{definition} \emph{(strong total order)}
\label{DEF/Strong Total Order}
Given a domain of values \(\set{D}\),
a \emph{strong total order} is a relation ``\(\prec\)''
which, \( \forall x, y, z \in\set{D}  \), is
    \begin{itemize}[nolistsep,leftmargin=*]
	\item \emph{transitive}.
		if \(x \prec y\) and \(y \prec z\), then \(x \prec z\),
	\item \emph{connex}.
		\(x \prec y\) or \(y \prec x\) 
		and
	\item \emph{irreflexive}
		\(x \not \prec x\).
    \end{itemize}
    
One may name a strong total order over \(\set{D}\)
as \Ord{T}[{\set{D}}][\prec].
Then
\(x \prec y \in\Ord{T}[{\set{D}}][\prec]\)
asks whether \(x\) precedes \(y\) in the order.
\end{definition}

We might know how \emph{groups} of values
are ordered,
but not 
how the values \emph{within} each group are ordered. This is a \emph{weak} total order.

\begin{definition}
\emph{(weak total order)}
\label{DEF/Weak Total Order}
Given a domain of values \(\set{D}\),
a \emph{weak total order} is a relation ``\(\prec\)''
defined over \set{D},
\Ord{W}[\set{D}][\prec],
\emph{iff} there is a partition over \set{D}'s values,
\(\coll{D} = \brac{\set{D}[1], \ldots, \set{D}[k]}\)
and a strong total order \Ord{T}[\coll{D}][\prec]\
over \coll{D}\
such that
\[
    \Ord{W}[\set{D}][\prec]
        =
    \{  \var{a} \prec \var{b}
            \,|\,
        \var{a}\in\set{D}[i]
            \land
        \var{b}\in\set{D}[j]
            \land
        \set{D}[i] \prec \set{D}[j] \in  \Ord{T}[\coll{D}][\prec]
    \}\mbox{.}
\]
\end{definition}

A strong partial order defines order precedence for some pairs
of items in the domain, but not all.

\begin{definition}
\emph{(strong partial order)}
\label{DEF/Strong Partial Order}
	Given a domain of values \(\set{D}\),
	a \emph{strong partial order} is a relation ``\(\prec\)''
	which \( \forall x, y, z \in\set{D}  \) is

    \begin{itemize}[nolistsep,leftmargin=*]
	\item \emph{antisymmetric}.
		if \(x \prec y\), then \(y \not\prec x\),
	\item \emph{transitive}.
		if \(x \prec y\) and \(y \prec z\), then \(x \prec y\), 
		and
	\item \emph{irreflexive}.
		\(x \not\prec x\).
    \end{itemize}

	One may name a strong partial order 
	over \(\set{D}\)
	as \Ord{P}[{\set{D}}][\prec].
	Then
	\(x \prec y \in\Ord{P}[{\set{D}}][\prec]\)
	asks whether \(x\) precedes \(y\) in the order.
\end{definition}

For any strong partial order \Ord{P}[{\set{D}}][\prec],
there exists a strong total order \Ord{T}[{\set{D}}][\prec]\
such that
\(
	\Ord{P}[{\set{D}}][\prec]
		\subseteq
	\Ord{T}[{\set{D}}][\prec]
\).
%
%
%
A \emph{nested} order %
(\emph{lexicographic} ordering)
with respect to a \emph{list} of attributes \blue{$\lst{X}$}
corresponds to the semantics of SQL's \emph{order by},  
\blue{shown as 
$\tup{s}\orel[\lst{X}]\tup{t}$
or 
$\tup{t}\orel[\lst{X}]\tup{s}$ 
between tuples $\tup{s}$ and $\tup{t}$}.
An attribute set can define a \emph{partition}
of a table's tuples into \emph{groups},
as by SQL's \emph{group by}.

\begin{table}[t]
\blue{
\paddingT
\centering
    \caption{\label{tab:notations} \blue{Notation}}
    \paddingD
    \renewcommand{\arraystretch}{1.05}
    \scalebox{0.87}{
\begin{tabular}{|l|l|}
\hline
\textbf{Notation} & \textbf{Description} \\ \hline \hline
$\set{D}$, $\coll{D}$ & Domain of ordered values, a partition\\ \hline
$\Ord{T}_{\set{D}}^{\prec}$, $\Ord{W}_{\set{D}}^{\prec}$, $\Ord{P}_{\set{D}}^{\prec}$ & Strong total, weak total, and strong partial order \\ \hline
$\R{R}$, $\T{r}$, $\A{A}$ & Relational schema, table instance, and one attribute \\ \hline
$\set{X}$, $\set{X}\set{Y}$, $\emptySet{}$ & Set of attributes, set union, and the empty set \\ \hline
$\lst{X}$, $\lst{X}'$ & List of attributes and arbitrary permutation  \\ \hline
$\lst{X}\lst{Y}$, $\emptyLst$ & List concatenation and empty list\\ \hline
$\tup{t}$, $\tup{t}_{\set{X}}$, $\tup{t}_{\lst{X}}$ & Tuple and projections 
over attributes (cast to set) \\ \hline
$\set{E}(\tup{t}_{\set{X}})$, $\pi_{\set{X}}$, $\tau_{\lst{X}}$ & Partition group, partition, and sorted partition \\ \hline
$\oc[\set{X}]{\A{A}}{\A{B}}$ & Canonical $\OC$ \\ \hline
$\od[\set{X}]{\emptyLst}{\A{A}}$ & Canonical $\OFD$ \\ \hline
$\Ord{T}_{\A{A}^*}^{\prec}$, $\A{A}^*$ & Derived and strongest derivable orders over $\A{A}$ \\ \hline
$\BG_{\A{A}, \A{B}}$, $\BGp_{\A{A}, \A{B}}$ & Initial and simplified bipartite graphs over $\A{A}$ and $\A{B}$ \\ \hline
\end{tabular}
}
}
\end{table}

\begin{definition}
\emph{(partition)}
\label{DEF/partition}
The \emph{partition group} of a tuple \(\tup{t} \in \T{r}\)
over an attribute set \(\set{X} \subset \R{R}\) is defined
as
\(
    \set{E}(\tup{t}_{\set{X}})
        =
    \{\tup{s} \in \T{r}\,|\,\proj{s}{\set{X}} = \proj{t}{\set{X}}\}
\).

The \emph{partition} \emph{of} \(\T{r}\) \emph{over} \(\set{X}\) is
the \emph{set} of partition groups
\(
    \pi_{\set{X}}
        =
    \{\set{E}(\tup{t}_{\set{X}})\,|\,\tup{t} \in \T{r}\}
\).
The \emph{sorted partition} \(\tau_{\lst{X}}\)
\emph{of} \(\T{r}\) \emph{over} \(\lst{X}\) is
the \emph{list} of partition groups
from \(\pi_{\set{X}}\)
\emph{sorted} with respect to \(\lst{X}\)
and the \emph{domain orders} of the attributes in \(\lst{X}\);
e.g.,
partition groups in 
\(\pi_{\brac{\A{A}, \A{B}, \A{C}}}\)
are sorted
in \(\tau_{\Lst{\A{A}, \A{B}, \A{C}}}\)
as per
SQL's ``order by \A{A}, \A{B}, \A{C}\@''\blue{\cite{TANE,SGG+2017:ODD}}.

\end{definition}

\begin{example} \label{example:partition-group}
In Table~\ref{table:festivals},
%
\blue{
   $\set{E}({\tup{t}[1]}_{\A{yearGreg}})=\{\tup{t}[1],$ $\dots,$ $\tup{t}[5]\}$,
    $\pi_{\A{yearGreg}}=\left\{\{\tup{t}[1], \dots, \tup{t}[5]\},
                \{\tup{t}[6], \dots, \tup{t}[10]\}
        \right\}$, and
    $\tau_{\A{yearGreg}}=[\{\tup{t}[1],\dots,\tup{t}[5]\},\{\tup{t}[6],\dots,\tup{t}[10]\}]$.
}
\end{example}

\subsection{Order Dependencies and Order Compatibility} \label{sec:prelim:ODs}
\noindent
\textbf{List-based notation}.
A natural way to describe an order dependency is
via two \emph{lists} of attributes.
An order dependency
\(\lst{X}\) \emph{orders} \(\lst{Y}\) means
that \(\lst{Y}\)'s values are
lexicographically, monotonically non-decreasing
with respect to \(\lst{X}\)'s values
\cite{%
	LN16,%
	Ng,%
	SGG+2017:ODD,%
	SGG2018:BiODD,%
	Szl2012
}.

\begin{definition}
\emph{(order dependency)}
\label{definition:OrderDependency}
Let $\set{X}, \set{Y} \subseteq \R{R}$.
$\orders{\lst{X}}{\lst{Y}}$ denotes
an \emph{order dependency} (\OD),
read as $\lst{X}$ \emph{orders} $\lst{Y}$.
\orders{\lst{X}}{\lst{Y}}\ (\(\T{r} \models \orders{\lst{X}}{\lst{Y}}\))
iff, \col{for all \(\tup{r}, \tup{s} \in \T{r}\),
\(\tup{r}\orel[\lst{X}]\tup{s}\) implies \(\tup{r}\orel[\lst{Y}]\tup{s}\).}
$\lst{X}$ and
$\lst{Y}$ are {\em order equivalent}, denoted as $\orderEquiv{\lst{X}}{\lst{Y}}$, iff $\orders{\lst{X}}{\lst{Y}}$ and $\orders{\lst{Y}}{\lst{X}}$.
\end{definition}

\begin{example}
The following \OD s hold in Table \ref{table:festivals}: \\
    \od{ \Lst{ \A{timestamp} } }
       { \Lst{ \A{year} } }\
	and
    \od{ \Lst{ \A{country}, \A{profit} } }
       { \Lst{ \A{country}, \A{tax} } }.
\end{example}


\begin{definition}
\emph{(order compatibility)}
\label{definition:orderCompatible}
Two lists $\lst{X}$ and $\lst{Y}$ are {\em order
compatible} ($\OC$), denoted as $\lst{X} \sim \lst{Y}$, iff
$\orderEquiv{\lst{XY}}{\lst{YX}}$.
\end{definition}


\begin{example}
Assume that we add an attribute 
\A{monthNum}\ as a numeric version of Gregorian month
in Table~\ref{table:festivals}.
Then, the \OC\
    \oc{\Lst{\A{yearGreg}, \A{monthNum}}}
       {\Lst{\A{yearGreg}, \A{week}}}\
is valid with respect to the table
as sorting by year, month and breaking ties by week
is equivalent to
sorting by year, week and breaking ties by month.
\end{example}


There is a strong relationship between $\OD$s and $\FD$s.
Any $\OD$ implies an $\FD$, modulo lists and sets, but not vice versa
\cite{Szl2012,SGG+2013:complexity}.
If $\R{R} \models \orders{\lst{X}}{\lst{Y}}$,
then $\R{R} \models \set{X} \rightarrow \set{Y}$
(\FD).
%
%
There also exists a correspondence between $\FD$s and
$\OD$s~\cite{Szl2012,SGG+2013:complexity}.
$\R{R}$ $\models$ $\set{X} \rightarrow \set{Y}$ ($\FD$\emph)
\emph{iff} $\R{R}$ $\models$
$\orders{\lst{X}'}{\lst{X}'\lst{Y}'}$.
%
%

$\OD$s can be violated in two ways~\cite{Szl2012,SGG+2013:complexity}.
$\R{R}$ $\models$ $\orders{\lst{X}}{\lst{Y}}$ iff $\R{R}$ $\models$
$\orders{\lst{X}}{\lst{XY}}$ \emph{(}$\FD$\emph{)} and $\R{R}$
$\models$ $\simular{\lst{X}}{\lst{Y}}$ ($\OC$).
%
%
This offers two sources of $\OD$ violations,
called 
	\emph{splits}
	and
	\emph{swaps}, respectively \blue{\cite{SGG+2017:ODD,Szl2012}}.


\begin{definition}
\emph{(split)}
\label{definition:split}
A \emph{split} with respect to
the \OD\ of \od{\lst{X}}{\lst{XY}}\
(which represents the \FD\ of \fd{\set{X}}{\set{Y}})
and table \T{r}\
is a pair of tuples
\(\tup{s}, \tup{t} \in \T{r}\)
such that
\(
    \tup{s}[\set{X}] = \tup{t}[\set{X}]
\)
but
\(
    \tup{s}[\set{Y}] \ne \tup{t}[\set{Y}]
\).
\end{definition}

\begin{definition}
\emph{(swap)}
\label{definition:swap}
A \emph{swap} with respect to
the \OC\ of \oc{\lst{X}}{\lst{Y}}\
and table \T{r}\
is a pair of tuples
\(\tup{s}, \tup{t} \in \T{r}\)
such that
\(
    \tup{s} \prec_{\lst{X}} \tup{t}
\)
but
\(
    \tup{t} \prec_{\lst{Y}} \tup{s}
\).
\end{definition}

\begin{example} \label{example:swap}
In Table~\ref{table:festivals},
there is a split for
the \OD\
\blue{candidate}
\od{\Lst{\A{yearGreg}}}{\Lst{\A{yearGreg}, \A{timestamp}}}\
\emph{(}an \FD\emph{)}
with tuples \tup{t}[1]\ and \tup{t}[2],
and a swap
for the \OC\
\blue{candidate}
\oc{\Lst{\A{count}}}{\Lst{\A{profit}}}\
with tuples \tup{t}[7]\ and \tup{t}[8]
\blue{, invalidating these candidates}.
\end{example}


\noindent
\textbf{Set-based notation (and mapping)}.
Expressing \OD s in a natural way relies on lists of attributes,
as in SQL order-by.
However, 
lists are \emph{not} inherently necessary
to express \OD s as we can express them in a set-based \emph{canonical} form.
The set-based form
enables more efficient \OD\ discovery,
and there exists a
polynomial \emph{mapping} of list-based \OD s
into \emph{equivalent} set-based canonical \OD s
\cite{SGG+2017:ODD,SGG2018:BiODD}.

\begin{definition}
\emph{(canonical form)}
\label{def:canonicalODs}
The \FD\ that states
that attribute \A{A}\ is \emph{constant}
with\-in each par\-ti\-tion group
over the set of attributes \set{X}\
can be written as
\od[\set{X}]{\emptyLst}{\A{A}}.
This is equivalent to the \OD\
\od{\lst{X}'}{\lst{X}'\A{A}}\
in list notation.
Call this an \emph{order} functional dependency (\OFD).
The \emph{canonical} \OC\ that states
that \A{A}\ and \A{B}\
are order compatible
within each partition group
over the set of attributes \set{X}\
is denoted as
\oc[\set{X}]{\A{A}}{\A{B}}.
This is equivalent to the \OC\
\(\simular{\lst{X}' \A{A}}{\lst{X}' \A{B} }\).

The set \set{X}\ in this notation is called
the \OFD's or \OC's \emph{context}.
\OFD s and canonical \OC s
constitute the canonical \OD s,
\blue{which we express using the notation $\ordersCtxSet{\set{X}}{\A{A}}{\A{B}}$.}

\end{definition}



\blue{We are 
interested in \OD s of the form
\od{\lst{X}'\A{A}}{\lst{X}'\A{B}} as written in
the canonical form
as \od[\set{X}]{\A{A}}{\A{B}}.
To discover such $\OD$s,}
we limit the search to find canonical \OC s and \OFD s.
This generalizes:
\od{\lst{X}}{\lst{Y}}\
\emph{iff}
\od{\lst{X}}{\lst{X}\lst{Y}}
\emph{and}
        \oc{\lst{X}}{\lst{Y}}.   
These can be encoded into an equivalent set
of \OC s and \OFD s~\cite{SGG+2017:ODD,SGG2018:BiODD}.
In the context of \set{X},
all attributes in \set{Y}\ are constants.
In the context
of all prefixes of \lst{X}\ and of \lst{Y},
the trailing attributes are order compatible.
Thus,
we can encode \od{\lst{X}}{\lst{Y}}\
based on the following polynomial mapping.

\begin{center}
\(
    \R{R} \models
        \od{\lst{X}}{\lst{X}\lst{Y}}
    \mathrel{\emph{iff}}
    \forall \A{A} \in \lst{Y}.\ 
    \R{R} \models
        \od[\set{X}]{\emptyLst}{\A{A}}
\)
and 
\end{center}


\begin{center}
\(
    \R{R} \models
        \oc{\lst{X}}{\lst{Y}}
    \mathrel{\emph{iff}}
\) 
\(
    \forall i,j.\ 
    \R{R} \models \oc[
                        \Lst{\A{X}_{1}, \ldots, \A{X}_{i-1}}
                        \Lst{\A{Y}_{1}, \ldots, \A{Y}_{j-1}}
                     ]
                     {\A{X}_{i}}
                     {\A{Y}_{j}}
\).
\end{center}

\noindent
This establishes a \emph{mapping} of list-based \OD s
into \emph{equivalent} set-based canonical \OD s; i.e.,
the $\OD$
\od{\lst{X}'\A{A}}{\lst{X}'\A{B}}\
is logically equivalent to the pair
of the \OC\ \oc[\set{X}]{\A{A}}{\A{B}}\
and \OFD\  \od[\set{X}\A{A}]{\emptyLst}{\A{B}}.
\blue{This is because $\lst{X}'$, which is a common prefix for both the left and right side of this $\OD$, can be factored out, making $\oc[\set{X}]{\A{A}}{\A{B}}$ the only non-trivial $\OC$ that needs to hold. Thus, $\OD \equiv \OC + \OFD$.}

\begin{example}\label{example:ODs}
	In Table~\ref{table:festivals},
	$\brac{ \A{country}, \A{profit} }$$:$
	$\orders{\emptyLst{}}{\A{tax}}$ \emph{(}\OFD\emph{)}
	and
	$\brac{\A{country}}$$:$
	$\simular{\A{profit}}{\A{tax}}$ \emph{(}$\OC$\emph{)}.
	Hence,
	$\brac{\A{country}}$$:$
	$\orders{\A{profit}}{\A{tax}}$ \emph{(}$\OD$\emph{)},
	as tax rates vary in different countries.
\end{example}

%
%
%



\noindent
\textbf{Lattice Traversal}.
\blue{
To discover $\OC$s, our algorithm starts with single attributes and proceeds to
larger attribute sets by traversing a lattice of all possible \emph{set}s of 
attributes in a level-wise manner. This search space is referred to as the lattice space.
}
When processing an attribute set $\set{X}$, the algorithm verifies (\textsf{O})\FD s of the
form
\od[\set{X} \setminus \A{A}]{\emptyLst}{\A{A}}\ 
for which \(\A{A} \in \set{X}\),
and \OC s of the form
\oc[\set{X} \setminus \brac{\A{A}, \A{B}}]{\A{A}}{\A{B}}\
for which \(\A{A}, \A{B} \in \set{X}\)
and \(\A{A} \ne \A{B}\).
The set-based \OD-discovery algorithm has
exponential worst-time complexity in the number of attributes (to generate the candidate $\OD$s), but linear complexity in the number of tuples (to verify each $\OD$ candidate)
\blue{\cite{SGG+2017:ODD}}.
That the canonical \OD s have a set-based representation rather than list-based means that the lattice is set-based, not list-based, 
making it significantly smaller.
%
%


%

\noindent
\textbf{Problem Statement}.
Given a dataset, 
we want to \emph{find implicit domain orders} by
extending the set-based \OD\ discovery
algorithm
\cite{%
        SGG+2017:ODD,
        SGG2018:BiODD
     }
to \EIOC s and \IIOC s
of the form
\oc[\set{X} \setminus \brac{\A{A}, \A{B}}]{\A{A}}{\A{B}[][*]}
and
\oc[\set{X} \setminus \brac{\A{A}, \A{B}}]{\A{A}[][*]}{\A{B}[][*]},
respectively.
Since  
\(\OD \equiv \OC + \FD\),
we also want to \emph{find implicit domain orders}
via \EIOD s and \IIOD s,
for which the \OFD\
\fd{\set{X}\A{A}}{\A{B}}\
(i.e., \od[\set{X}\A{A}]{\emptyLst{}}{\A{B}})
additionally holds.
\

\subsection{\blue{Discovery Framework}} \label{sec:prelim:framework}
\begin{change}
Figure~\ref{fig:framework} illustrates the framework of our algorithm ($\iOrder$). First, potential $\OC$ candidates are generated for one level of the lattice.
These candidates are then pruned using the dependencies found in the previous levels of the lattice and the $\OC$ axioms.
Next, the existence of an $\FD$ is checked for each candidate, and depending on whether an $\FD$ holds or not, different types of implicit $\OC$s are examined 
using the algorithms described in Sections~\ref{sec:eioc} through \ref{SEC:SAT-Solver}.
Next, valid candidates and the strongest implicit $\OC$ types that were validated 
are stored (e.g., unconditional $\EIOC$s are preferred over conditional $\EIOC$s and unconditional $\IIOC$s). 
The candidates for the next level of the lattice are then generated, until the search for candidates is finished.
In the final step, the discovered implicit orders are ranked based on their interestingness scores (Section~\ref{sec:measure-interestingness-def}).
\end{change}

\begin{figure}[t]
    \center
    \scalebox{0.95}{
    \includegraphics[width=0.45\textwidth]{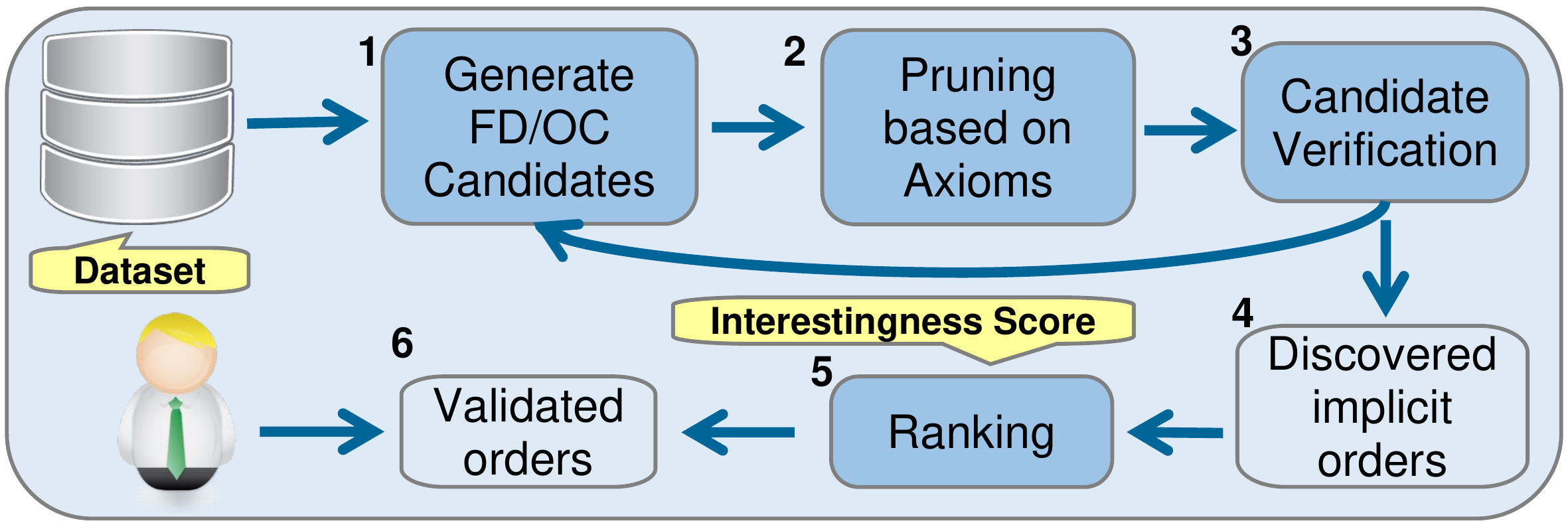}
    }
        \paddingT
     \caption{\blue{System framework.}}
     \label{fig:framework}
      \paddingD
\end{figure}

\section{E/I Discovery} \label{sec:eioc}

We begin with 
explicit-to-implicit ({\sf{E/I}}) domain order discovery
through \OD s and \OC s
in which an attribute with an explicit order
is used to find
an implicit domain order over another attribute.

\begin{itemize}[nolistsep,leftmargin=*]
\item
    We first define 
    an \emph{implicit} domain order
    with respect to the table
    and an explicit domain order on another attribute
    (Section~\ref{sec:implicit orders}).
    
\item
    We then subdivide the problem of domain-order discovery
    via \EIOC s and \OD s
    as follows:
    
    \begin{itemize}[nolistsep,leftmargin=*]
    \item with \FD s,
    	thus effectively via \OD s
    	(Section \ref{sec:eioc-with-fd});
    \item without corresponding \FD s,
    	thus effectively via \OC s
    	\begin{itemize}[nolistsep,leftmargin=*]
    	\item with \emph{empty} contexts
    		(Section~\ref{sec:eioc-empty-context})
    		and
    	\item with \emph{non-empty} contexts
    		(Section~\ref{section:nonemptyEI}).
    	\end{itemize}
    \end{itemize}
\end{itemize}

\begin{change}
 Thus,
in Section~\ref{sec:implicit orders},
we define 
when two attributes 
can be \emph{co-ordered},
given an explicit 
order on one, and define 
what the \emph{strongest} derived order is. 
We then provide
\emph{algorithms}
to determine when \oc[\set{X}]{\A{A}}{\A{B}[][*]},
and to compute the \emph{strongest} order \A{B}[][*]\
when it does.
\end{change}




\subsection{Implicit Domain Orders} \label{sec:implicit orders}

For \emph{explicit-explicit} \OC\ discovery,
say, for columns \A{A}\ and \A{B},
it suffices to check that the tuples of \T{r}\
can be ordered in \emph{some} way
that is consistent \emph{both}
with ordering the tuples of \T{r}\
with respect to column \A{A}'s explicit domain order
\emph{and}
with respect to column \A{B}'s explicit domain order.
That way of ordering the tuples of \T{r}\
is a \emph{witness} that \A{A}\ and \A{B}\ can be ``co-ordered'';
it \emph{justifies} that \oc{\A{A}}{\A{B}}.

To define \emph{explicit-implicit order compatibility},
we want to maintain this same concept:
there is a way to order the tuples of \T{r}\
with respect to column \A{A}'s explicit domain order
\emph{and} for which the projection on \A{B}\ provides
a \emph{valid} order over \A{B}'s values.
For \EIOC s with a non-empty context,
\oc[\set{X}]{\A{A}}{\A{B}[][*]},
there must be a witness total order over \T{r}\
that is,
\emph{within} each partition group of \set{X},
compatible with the explicit order of \A{A}\
\emph{and}
the order over \A{B}\ dictated by this is valid.
%
This answers one of our two questions:
whether the candidate \OC\ of \oc{\A{A}}{\A{B}[][*]}\ \emph{holds}
over \T{r}.
The second question in this case, though,
that we also need to answer is,
what is \emph{that} order \A{B}[][*]?


Such a witness order over \T{r}\
derives a 
total order (perhaps weak) over \A{B}.
There may be more than one witness order over \T{r}.
Consider the $\OC$ $\simular{\A{monthNum}}{\A{monthLun}^*}$ over the first five tuples in Table~\ref{table:festivals}. While the ordering $[\tup{t}_1, \tup{t}_2, \tup{t}_3, \tup{t}_4, \tup{t}_5]$ is a valid witness that gives the order $Corner\prec Peach\prec Plum\prec Pomegranate\prec Winter$ over $\A{monthLun}$, so is the ordering $[\tup{t}_1, \tup{t}_3, \tup{t}_2, \tup{t}_4, \tup{t}_5]$, where the order of month values $Peach$ and $Plum$ is swapped. This indicates that we can only derive a \emph{weak} total order $Corner\prec\{Peach, Plum\}\prec Pomegranate\prec Winter$. 

\blue{More formally, if}
\oc[\set{X}]{\A{A}}{\A{B}[][*]}\ holds over \T{r},
\blue{we define the 
\emph{strongest derivable} order
\A{B}[][*] as}
the intersection of all the derived strong total orders over \A{B}\
corresponding to the possible witness total orders of \T{r}\
that justify \oc[\set{X}]{\A{A}}{\A{B}[][*]}.
\blue{We shall present algorithms for validating an $\EIOC$ candidate $\oc[\set{X}]{\A{A}}{\A{B}[][*]}$ and deriving the order $\A{B}^*$.}

\begin{figure}[t]
    \centering
    \begin{subfigure}[]{0.23\textwidth} 
        \centering
        \includegraphics[height=5cm]{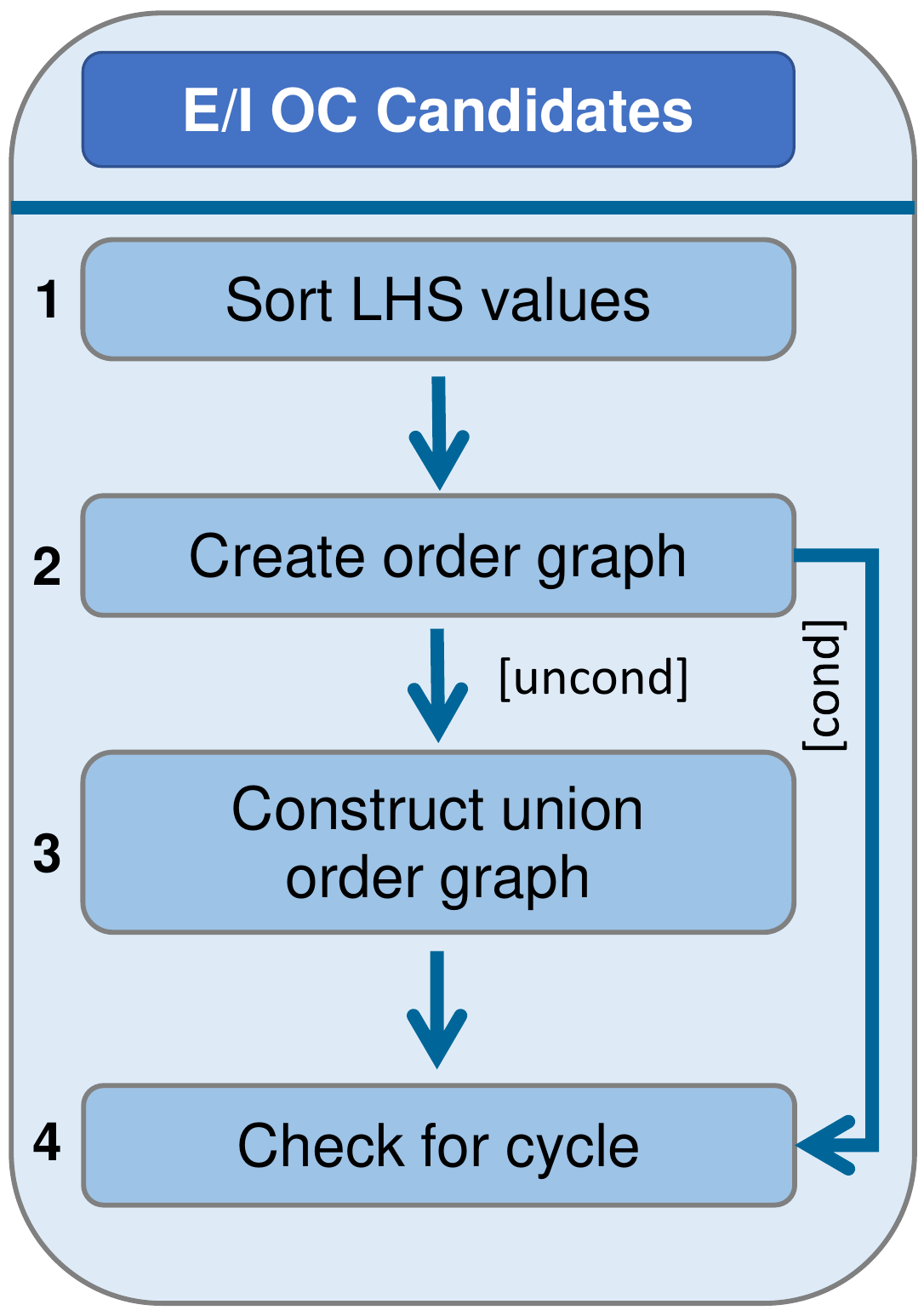}
        \caption{\label{fig:diagram-ei} \EIOC\ case.}
    \end{subfigure}%
    ~ 
    \begin{subfigure}[]{0.23\textwidth} 
        \centering
        \includegraphics[height=5cm]{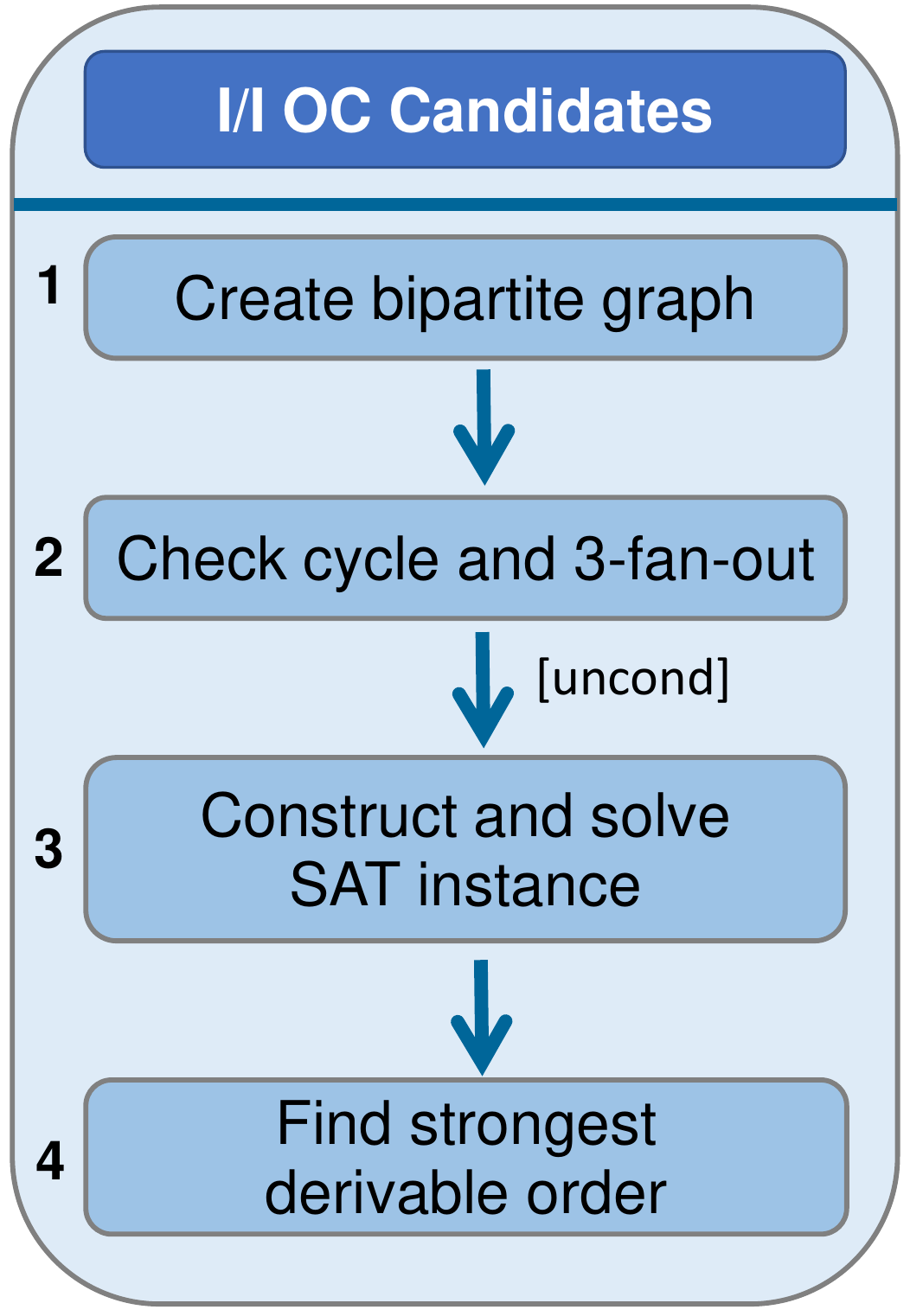}
        \caption{\label{fig:diagram-ii} \IIOC\ case.}
    \end{subfigure}
    \paddingT
    \caption{Discovery process.}
     \paddingD
    \label{fig:diagrams}
\end{figure}

\begin{example}
The \EIOC\ of \oc{\A{yearLun}}{\A{monthLun}[][*]}\
holds over the first five tuples
in Table~\ref{table:festivals}.
There is a single value for 
\A{yearLun}.
That there exists a total order
over the tuples with no cycle
over the values of \A{monthLun}\ (``\A{B}'')
means then that there is a valid witness.
However,
since \emph{all} such total orders do not cause cycles, 
we have no information regarding
the order between month values. 
Thus,
\(\A{monthLun}[][*] = \emptySet\),
the \emph{empty} partial order.
\end{example}





\subsection{E/I OD, \blue{Empty Context}} \label{sec:eioc-with-fd}

\blue{We first consider $\EIOD$s with an empty context; i.e., we are looking to establish
whether there is a \A{B}[][*]\
with respect to \A{A}\
over the whole table $\T{r}$ when}
there is a functional dependency
from one side to the other  
of a candidate.
When we have an explicit order over one side,
we might discover an implicit order over the other side
by finding an \EIOD\ between them.

Let our pair of attributes be \A{A}\ and \A{B},
assume we have an explicit order over \A{A},
and
we want to discover an implicit order over \A{B}\
(i.e., \A{B}[][*]).
We have three cases for $\FD$s between the pair:
%
	(1) \fd{\A{A}}{\A{B}}\
	and \fd{\A{B}}{\A{A}};
	(2) \fd{\A{A}}{\A{B}}\
	but \fdN{\A{B}}{\A{A}};
	or
	(3) \fd{\A{B}}{\A{A}}
	but \fdN{\A{A}}{\A{B}}.
\blue{We devise efficient algorithms for each case.}


The first case above is trivial.
There exists exactly \emph{one}
implicit order over \A{B},
which is a strong total order.
%
\blue{To discover this order over $\A{B}$,}
sort the table over \A{A},
and
project out \A{B}.
(If \A{A}\ is not a key of the table and may have duplicates,
then \A{B}\ would too;
eliminate these duplicates, which must be adjacent.)
This is \A{B}[][*].
This is unique
with respect to \A{A}\
and is a strong total order.

\begin{example}
Let the attribute \A{monthNum}\ be added
to  Table~\ref{table:festivals}
to denote the Gregorian month of the year in the numeric format.
Then
the \FD s
\fd{\A{monthNum}}{\A{monthGreg}}\
and 
\fd{\A{monthGreg}}{\A{monthNum}}\ hold.
Thus,
the \EIOC\ of
\oc{\A{monthNum}}{\A{monthGreg}[][*]}\ is valid
with the implied domain order
\Ord{P}[{\A{monthGreg}[][*]}][\prec]\ of
\(
    \emph{Jan}
        \prec
    \emph{Apr}
        \prec
    \emph{Jun}
        \prec
    \emph{Aug}
        \prec
    \emph{Oct}
        \prec
    \emph{Dec}
\).
\end{example}

For the second case,
since \fdN{\A{B}}{\A{A}},
this means some \A{B}\ values are associated
with more than one \A{A}\ value.
%
We can partition the tuples of \T{r}\ by \A{B}.
This can be done in \({\mathcal O}(|\T{r}|)\)
via a hash.
Scanning the partition,
we find the \emph{minimum} and \emph{maximum} values
of \A{A}\ within each \A{B}-value group.
Then the \A{B}-value partition groups are
sorted by their associated min-\A{A}'s. 
If
\(|\A{B}| \ll |\A{A}|\ (\approx |\T{r}|)\),
this is less expensive than sorting by \A{A}.
If the intervals of the values of $\A{A}$ co-occurring with each value of $\A{B}$ do not \emph{overlap}, 
then this is \A{B}[][*].
To formalize this,
let us define the notion of an \emph{interval partitioning}.


\begin{definition}
Let
\(
	(\pi_{\A{B}})_{\A{A}}
		=  
    \{\set{E}_{1}(\tup{t}_{\A{B}})_{\A{A}},
	  \set{E}_{2}(\tup{t}_{\A{B}})_{\A{A}}, 
		  \ldots,
      \set{E}_{k}(\tup{t}_{\A{B}})_{\A{A}}
	\}
\)
be the \emph{partitioning of} an attribute \A{A}\
\emph{by} an attribute \A{B}.
Call the partitioning an \emph{interval partitioning}
\emph{iff} there does not exist \(i, j \in [1,\ldots,k]\)
such that $i<j$ and 
\(
	\mbox{min}(\set{E}_{i}(\tup{t}_{\A{B}})_{\A{A}})
		\prec
	\mbox{min}(\set{E}_{j}(\tup{t}_{\A{B}})_{\A{A}})
		\prec
	\mbox{max}(\set{E}_{i}(\tup{t}_{\A{B}})_{\A{A}})
\).
\end{definition}

An interval partitioning allows us to
separate the \emph{ranges} of \A{A}\ values w.\ r.\ t.\ the \A{B}\ values
such that the ranges do not overlap.


\begin{theorem}
\blue{Assume the \FD\ of \fd{\A{A}}{\A{B}}\ holds in \T{r}.
Then
\(\simular{\A{A}}{\A{B}^{*}}\) holds
\emph{iff}
\((\pi_{\A{B}})_{\A{A}}\)
is an interval partitioning;
\Ord{T}[{\A{B}_{*}}][\prec]\
is the unique order of \A{B}'s values
corresponding to their order
in \T{r}\ as sorted by \A{A}.}
\end{theorem} 

\begin{example}
Consider
attributes 
$\A{count}$ and $\A{size}$
in Table~\ref{table:festivals}.
The $\FD$ $\fd{\A{count}}{\A{size}}$ holds;
\((\pi_{\A{size}})_{\A{count}}\) is
an interval partitioning
with
\(
	\tau_{\A{count}}
		=
	[ \tup{t}[10],
	  \tup{t}[9],
	  \tup{t}[3],
	  \tup{t}[7],
	  \tup{t}[8],
	  \tup{t}[2],
	  \tup{t}[6],
	  \tup{t}[5],
	  \tup{t}[4],
	  \tup{t}[1] ]
\).
Thus,
the implied domain order
\Ord{T}[{\A{size}_{*}}][\prec]\ is
\(
	\emph{Small}
	\prec
	\emph{Medium}
	\prec
	\emph{Large}
	\prec
	\emph{X-Large}
\), as per the \OC\
\oc{\A{count}}{\A{size}[][*]};
i.e, the  \OD\ of
\od{\A{count}}{\A{size}[][*]}\ holds.
\end{example}

The third case looks like the second case,
\emph{except} the explicit order known
over \A{A}\ is now
on the right-hand side
of our \FD, \fd{\A{B}}{\A{A}}. 
We can take a similar interval-partitioning approach as before. 
\emph{Sort} the table \T{r}\ by \A{A}.
If \(|\A{A}| \ll |\T{r}|\),
this is more efficient than fully sorting \T{r}.
This computes the sorted partition
\(\tau_{\A{A}}\).
The \A{A}\ values \emph{partition} the \A{B}\ values,
since \fd{\A{B}}{\A{A}}
and \(\tau_{\A{A}}\) orders these \emph{groups} of \A{B}\ values.
Since there are multiple \A{B}\ values
in some of the partition groups
of \(\tau_{\A{A}}\),
given that
\fdN{\A{A}}{\A{B}},
this does \emph{not}
determine an order
over \A{B}\ values
within the same group.
Thus
the \A{B}[][*]\ implied
by \(\tau_{\A{A}}\)
is not a \emph{strong} total order,
but it is a \emph{weak} total order. 

\begin{example}
Let the attribute \A{quarter}\ be added
to  Table~\ref{table:festivals}
to denote the year quarter;
i.e., \emph{Q1}, \emph{Q2}, \emph{Q3}, and \emph{Q4}.
The \FD\ of
\fd{\A{monthGreg}}{\A{quarter}}\ holds
as the Gregorian months perfectly align with the quarters.
Thus,
the \EIOC\
\oc{\A{quarter}}{\A{monthGreg}[][*]}\ holds;
\(\A{monthGreg}[][\tau_{\A{quarter}}]\)
is a weak total order
with 
    \brac{\emph{January}}
        $\prec$
    \brac{\emph{April}, \emph{June}}
        $\prec$
    \brac{\emph{August}}
        $\prec$
    \brac{\emph{October}, \emph{December}}.
Between months within each quarter,
we cannot infer any order.
\end{example}

\noindent
Let 
the $\FD$ be 
\fd{\A{A}}{\A{B}},
\(m = |\A{B}|\) (the number of distinct values of \A{B}),
and
\(n = |\T{r}|\) (the number of tuples).
In practice,
it is common that \(m \ll n\).

\begin{lemma}
\blue{The runtime of discovering $\EIOD$s with empty 
context is ${\mathcal O}(m\ln m + n)$.}
\end{lemma}

Proofs and pseudocode can be found in the appendix in Section~\ref{sec:appendix}.

\subsection{E/I OC, Empty Context} \label{sec:eioc-empty-context}

We next consider \EIOC s
with an empty context
in the form of
\(\simular{\A{A}}{\A{B}[][*]}\).
\blue{Similar to the previous section, the goal is to verify whether there is 
an order over the values of $\A{B}$ with respect to the order over the values of $\A{A}$. 
Using the sorted partitions of $\tau_{\A{A}}$, we
infer the order $b_i$ $\prec$ $b_j$ for every two distinct values of $\A{B}$ 
which co-occur with two \emph{consecutive} partition groups of $\A{A}$.
Let $\A{B}[][{\tau_{\A{A}}}]$ denote the set of these inferred relations over $\A{B}$.
We next check whether $\A{B}[][{\tau_{\A{A}}}]$ is a valid 
weak total
order;
if it is so, $\simular{\A{A}}{\A{B}[][*]}$ is a valid $\EIOC$ and 
$\A{B}[][*] \equiv \A{B}[][{\tau_{\A{A}}}]$.}
\diagram{Figure~\ref{fig:diagram-ei} demonstrates these steps.}

\begin{theorem}\label{thm:emptyContEI}
\blue{$\simular{\A{A}}{\A{B}[][*]}$ is 
valid
\emph{iff}
$\A{B}[][{\tau_{\A{A}}}]$ is a 
weak total
order.}
\end{theorem}

\begin{example}
\label{EX/monthNumLunEmpty}
The \EIOC\
\(\simular{\A{monthNum}}{\A{monthLun}^*}\)
does not hold
in Table~\ref{table:festivals} 
%
since the lunar month \emph{Winter}
co-occurs both
with December and January,
which have numerical ranks of $6$ and $1$ in the table,
\blue{resulting in \A{monthLun}[][{\tau_{\A{monthNum}}}] being invalid.}

\end{example}





\blue{In the following sections, $n$, $m$, and $p$ denote
the number of tuples, the number of distinct values of the candidate attribute(s) with 
implicit order, and the number of partition groups of the context.}

\begin{lemma}
\blue{The runtime of discovering $\EIOC$s with empty context is ${\mathcal O}(n + m^2)$, given an initial sorting of the values in the first level of the lattice.}
\end{lemma}






\subsection{\blue{E/I OD and} E/I OC, Non-empty Context} \label{section:nonemptyEI}

\begin{figure}[t]
    \centering
    \begin{subfigure}[]{0.2\textwidth} 
        \centering
        \includegraphics[height=3.5cm]{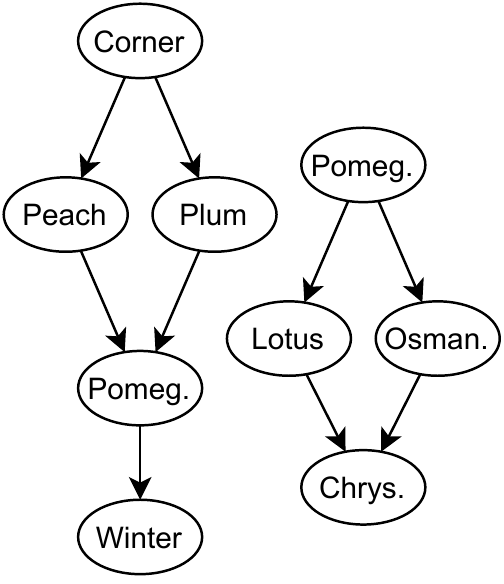}
        \caption{\label{fig:eioc-num-lunar-2-ec} Order in partition groups.}
    \end{subfigure}%
    ~ 
    \begin{subfigure}[]{0.2\textwidth} 
        \centering
        \includegraphics[height=3.5cm]{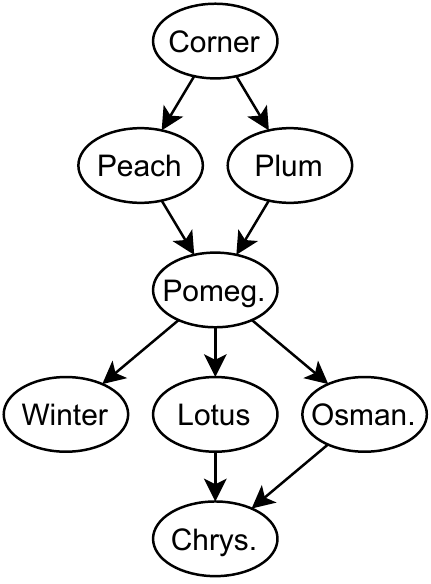}
        \caption{\label{fig:eioc-num-lunar-merged} Union order.}
    \end{subfigure}
    \paddingT
    \caption{Partial and union orders.}
     \paddingD
    \label{fig:GPMonths}
    
\end{figure}

When the context is non-empty,
say \set{X},
we first consider each partition group
in \(\pi_{\set{X}}\) independently.
This is equivalent,
with respect to each partition group,
to the empty-context $\EIOD$ and $\EIOC$ discoveries above.
\blue{For a candidate $\oc[\set{X}]{\A{A}}{\A{B}[][*]}$, if either of the $\FD$s $\fd{\set{X}\A{A}}{\A{B}}$ or $\fd{\set{X}\A{B}}{\A{A}}$ hold, we can use the algorithms in Section~\ref{sec:eioc-with-fd}.
Otherwise, we use the approach in Section~\ref{sec:eioc-empty-context}.
}
If an implicit order is discovered within \emph{each} partition group,
\blue{then the \emph{conditional} $\EIOC$ (or $\EIOD$) holds. To verify the unconditional case,}
we take the \emph{union} of those orders%
---%
each of which represents a weak 
total order%
---%
\blue{by including the edge $(a, b)$ in the \emph{union graph} \emph{iff} this edge exists in at least one of the individual orders}, 
and test whether this \emph{union graph} represents
a strong partial order
(i.e., is cycle free).
If so,
we have established an \blue{unconditional} \A{B}[][*]\
in the context of \set{X}.
\diagram{These steps are included in Figure~\ref{fig:diagram-ei}.
Note the third step on constructing the union order graph, which is not necessary 
for conditional $\EIOC$s.}



\begin{theorem}
There exists an implicit domain order
\(\Ord{P}[{\A{B}[][*]}][\prec]\)
such that the \EIOC\
\(\simularCtxSet{\set{X}}{\A{A}}{\A{B}[][*]}\) holds
\emph{iff}
the union graph is cycle free.
\end{theorem}


\begin{example} \label{example:union-graph}
The \EIOC\ 
\oc[\brac{\A{yearGreg}, \A{yearLunar}}]
   {\A{monthNum}}
   {\A{monthLun}[][{\tau_{\A{monthNum}}}]}
\noindent
holds \blue{\emph{unconditionally}} in Table~\ref{table:festivals}
since the union graph is cycle free. 
%
Figure~\ref{fig:eioc-num-lunar-2-ec}
shows the partial orders
corresponding to this \EIOC\ 
for years (2020, 4718) and (2021, 4719)
\blue{, each derived from one partition group using the algorithm described in Section~\ref{sec:eioc-empty-context}}.
(Note that the partition group for years (2021, 4718) is ignored since it only has one tuple.)
Figure~\ref{fig:eioc-num-lunar-merged} shows the union order.
\blue{Note that in the resulting union order, an edge is included \emph{iff} it belongs to at least one of the order graphs in Figure~\ref{fig:eioc-num-lunar-2-ec}; e.g., the edge $(Pomeg., Lotus)$ is included while $(Lotus, Osman)$ is not.}
\end{example}

\begin{example}
In Table~\ref{table:festivals},
the \FD\
\(\fd{ \A{country}, \A{count} }{\A{ribbon}}\)
holds. 
\blue{Given the $\EIOD$ candidate $\simularC{\A{country}}{\A{count}}{\A{ribbon}^*}$,} 
\((\pi_{\A{ribbon}})_{\A{count}}\) is
an interval partitioning
within each partition group
with respect to the context. 
However,
the candidate implicit orders 
over \A{ribbon}%
---%
\(\emph{White} \prec \emph{Red} \prec \emph{Blue}\)
within \emph{China} 
and 
\(\emph{White} \prec \emph{Blue} \prec \emph{Red}\)
within \emph{Canada}
---%
are not consistent,
as the \emph{Blue} and \emph{Red} values are flipped, \blue{making this candidate hold only \emph{conditionally}}.
\end{example}

Building the graph data-structure
for the union of the group orders (DAGs)
is simple.
\blue{This can be done by traversing the order from each partition group and adding each of their edges to the final graph if they have not been added yet.}
We then walk the resulting graph
by \emph{depth-first search} (DFS)
to determine whether it is cycle free.

\begin{lemma}
\blue{The time complexity of discovering $\EIOC$s with non-empty context is ${\mathcal O}(n\ln n+pm^2)$.}
\end{lemma}
\section{I/I Discovery} \label{sec:iioc}





A surprise for us
was 
that domain orders can also be discovered
even when \emph{no} explicit domain orders are known!

\begin{itemize}[nolistsep,leftmargin=*]
\item 
    We first must extend what is meant
    by an implicit domain order
    as defined in Section~\ref{sec:implicit orders}:
    now it is \emph{two} implicit domain orders
    that we seek to discover 
    (Section~\ref{sec:pairs of implicit domains}).

\item 
    We then subdivide the problem
    of domain-order discovery
    via \IIOC s and \OD s
    as follows:
    
    \begin{itemize}[nolistsep,leftmargin=*]
    \item 
    	candidates that have an empty context
    	\emph{or} that have a non-empty context
    	that is treated as \emph{conditional}
    	(Section~\ref{sec:iioc-empty-context});
    \item 
    	candidates that have a non-empty context
    	that is treated as \emph{unconditional}
    	(Section~\ref{ii-with-context});
    	and
    \item 
    	that have a corresponding \FD\
    	(Section~\ref{SEC/IIODs}).
    \end{itemize}
\end{itemize}

\begin{change}
Thus,
in Section~\ref{sec:pairs of implicit domains},
we define when two attributes,
\A{A}\ and \A{B},
with a context \set{X}\
can be co-ordered.
We also define 
what \emph{strongest} orders can be derived;
i.e., 
\A{A}[][*]\ and \A{B}[][*]. 
The following sections then provide
\emph{algorithms}
to determine when \oc[\set{X}]{\A{A}[][*]}{\A{B}[][*]},
and to compute the \emph{strongest} orders
\A{A}[][*]\ and \A{B}[][*]\
when it does.   
\end{change}

\subsection{Pairs of Implicit Domain Orders}
\label{sec:pairs of implicit domains}

As in the explicit-implicit case,
we have \emph{two} questions to address:
when does \oc[\set{X}]{\A{A}[][*]}{\A{B}[][*]}\ hold over \T{r};
and, if it does,
what are \emph{strongest} partial orders
that we can derive
for \A{A}[][*]\
and
\A{B}[][*].
Our criterion for whether \oc[\set{X}]{\A{A}[][*]}{\A{B}[][*]}\ holds over \T{r}\ 
is the same as before:
there exists \emph{some} strong total order \Ord{T}[\T{r}][\prec]\
over the tuples in \T{r},
a witness,
such that \A{A}\ and \A{B}'s values projected
into lists from \T{r}\ ordered thusly
represent strong total orders over $\A{A}$ and $\A{B}$'s values, respectively.


To determine the strongest derivable orders
for \A{A}[][*]\ and \A{B}[][*]\
is not the same as before,
however.
We cannot define it as simply, 
as the intersection of \emph{all} the projected orders.
The 
reason is
that there is never a \emph{single} witness;
witnesses come in pairs.
Since we have no explicit order to anchor the choice,
if we have a strong total order on \T{r}\ that is a witness,
then the \emph{reversal} of that order is also a witness.
%
%
Which direction,
``ascending'' or ``descending'',
is the right one to choose?
The choice is arbitrary.
We call this \emph{polarization}.

\begin{example} \label{example:polarity-size-color}
\blue{Consider the $\IIOC$ candidate $\simular{\A{size}}{\A{color}}$ and the first five rows of Table~\ref{table:festivals}. The total order $\Ord{T}[\T{r}][\prec]$ of $\tup{t}_3$ $\prec$ $\tup{t}_2$ $\prec$ $\tup{t}_5$ $\prec$ $\tup{t}_4$ $\prec$ $\tup{t}_1$ is a valid \emph{witness order}, and results in the \emph{derived order}s \emph{Medium} $\prec$ \emph{Large} $\prec$ \emph{X-Large} and \emph{White} $\prec$ \emph{Red} $\prec$ \emph{Blue} over the values of $\A{size}$ and $\A{color}$, respectively. However, the reverse of $\Ord{T}[\T{r}][\prec]$ with the opposite \emph{polarization} (i.e., $\tup{t}_1$ $\prec$ $\tup{t}_4$ $\prec$ $\tup{t}_5$ $\prec$ $\tup{t}_2$ $\prec$ $\tup{t}_3$) is also a valid witness, resulting in reversed derived orders over $\A{size}$ and $\A{color}$ as well. Through $\IIOC$s, it is not possible to assert which one of these polarizations is the correct one.}
\end{example}

To circumvent
that a witness order over \T{r}\
and the reversal of that order
which is also a witness
``cancel'' each other out
(that is,
their intersection is the empty order),
we define a \emph{witness class},
which consists of all the witnesses
that are reachable from one another via \emph{transpositions}.
That a transposition leads to another order which is also a witness
means that the forced order between tuples
in the two transposed blocks
is immaterial to how \A{A}\ and \A{B}\ can be co-ordered.

\blue{More formally, by applying a valid \emph{transposition} to a \emph{witness} order $\Ord{T}[\set{D}][\prec]$ over $\T{r}$, we can move a contiguous sub-sequence of tuples to another location within the same partition group, subject to the condition that the new total order $\Ord{S}[\set{D}][\prec]$ over tuples has to be a valid \emph{witness} as well. We refer to the set of all witness total orders over $\T{r}$ that can be converted to each other via one or more valid transpositions as a \emph{witness class}. Given a witness total order $\Ord{T}[\set{D}][\prec]$, the \emph{strongest derivable} order pairs $\A{A}^*$ and $\A{B}^*$ can be defined as the \emph{intersection} of all of the orders within the witness class of $\Ord{T}[\set{D}][\prec]$.}

\subsection{I/I OCs, Empty-Context or Conditional}
\label{sec:iioc-empty-context}

We first consider the cases of
\IIOC s with an empty context
and when the context is not empty
but for which
we treat the partition groups as independent (\emph{conditional}).
\blue{
For the $\IIOC$ candidate $\oc[\set{X}]{\A{A}[][*]}{\A{B}[][*]}$, 
our goal is to discover whether, 
within each partition group,
there exist
$\A{A}[][*]$ and $\A{B}[][*]$
such that they can be co-ordered.
To do so, 
}
we build a bipartite graph,
\BG[\A{A}, \A{B}]\ over 
\T{r}.
In this,
the nodes on the left represent
the partition groups by \A{A}'s values in \T{r},
\(\pi_{\set{X}\A{A}}\),
and those on the right represent
the partition groups by \A{B}'s values in \T{r},
\(\pi_{\set{X}\A{B}}\).
For each tuple \(\tup{t}\in\T{r}\),
there is an edge between
\tup{t}[\A{A}]\ 
(left)
and
\tup{t}[\A{B}]\ 
(right).

\begin{definition}
\emph{(3-fan-out)}
\label{DEF/3-fan-out}
	A bipartite graph has a \emph{3-fan-out}
	\emph{iff} it has a node
	that is connected to at least three other nodes.
\end{definition}




\begin{figure}[t]
    \centering
    \begin{subfigure}[]{0.15\textwidth} 
        \centering
        \scalebox{0.91}{
        \includegraphics[width=2.2cm]{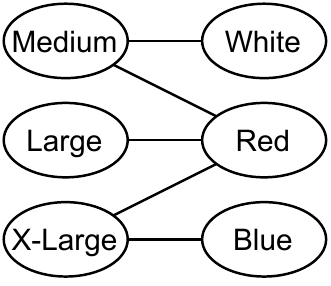}
        }
        \caption{China.}\label{fig:size-ribbon-china}
    \end{subfigure}%
    ~ 
    \begin{subfigure}[]{0.15\textwidth} 
        \centering
        \scalebox{0.91}{
        \includegraphics[width=2.2cm]{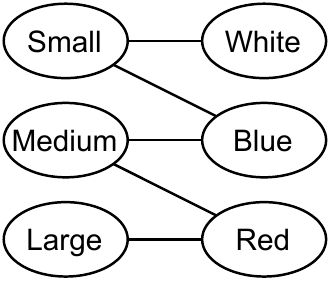}
        }
        \caption{Canada.}\label{fig:size-ribbon-canada}
    \end{subfigure}
    ~ 
    \begin{subfigure}[]{0.15\textwidth} 
        \centering
        \scalebox{0.91}{
        \includegraphics[width=2.2cm]{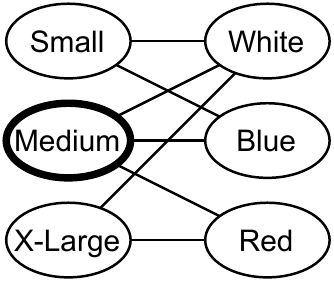}
        }
        \caption{All tuples.}\label{fig:size-ribbon-all}
    \end{subfigure}
    \paddingT
    \caption{\sf{BG}s or \sf{BG'}s for \sf{I/I OC}.}
     \paddingD
    \label{fig:size-ribbon}
\end{figure}

It does not suffice to consider directly
\BG[\A{A}, \A{B}]\
to determine whether
\oc{\A{A}[][*]}{\A{B}[][*]}.
This is because a node of degree one
in the \BG\ over \T{r}\
can never invalidate
the \IIOC.
E.g.,
\emph{White} has just degree one
in both of the \BG s
in Figures~\ref{fig:size-ribbon-china} and \ref{fig:size-ribbon-canada}.
These have to be excluded before
we check the \emph{3-fan-out} property.

\begin{definition}
\emph{(Singletons and \BGp)}
\label{DEF/Singleton}
	Call a node in a \BG\ with degree one a \emph{singleton}.
	Let \BGp\ be the \BG\
	in which the singletons
	and their associated edges
	have been removed.
\end{definition}

With \BGp[\A{A}, \A{B}],
we can test whether \oc{\A{A}[][*]}{\A{B}[][*]}.

\begin{theorem} \label{theo:bipValidity}
\oc{\A{A}[][*]}{\A{B}[][*]}\
is \emph{valid}
over \T{r}\
\emph{iff}
both of the following two conditions are \emph{true}
for \BGp[\A{A}, \A{B}]\ over \T{r}:

\begin{enumerate}[nolistsep,leftmargin=*]
\item it contains no 3-fan-out;
	and
\item it is acyclic.
\end{enumerate}
\end{theorem}

The intuition behind the requirement
of no 3-fan-outs
is that there has to be a way
to order the left values in an attribute on the left
to order the right values in an attribute on the right
such that none of the edges of \BGp\ \emph{cross}.
Also, there is no order if there is a cycle.
\diagram{Figure~\ref{fig:diagram-ii} demonstrates these steps, 
where only the first two steps are required for conditional $\IIOC$ candidates.}

\begin{example}
The \BGp[\A{size}, \A{ribbon}]\
over Table~\ref{table:festivals}
and shown in Figure~\ref{fig:size-ribbon-all}
has 3-fan-out:
\emph{Medium} connects to
\emph{White},
\emph{Blue},
and
\emph{Red}.
Thus,
the candidate \IIOC\ of
\oc[\emptySet]{\A{size}[][*]}{\A{ribbon}[][*]}\
is not valid.
\end{example}

\blue{Even though the $\IIOC$ candidate}
\oc[\emptySet]{\A{size}[][*]}{\A{ribbon}[][*]}\
over Table~\ref{table:festivals}
does not hold%
\blue{, it}
does not mean that
\oc[\set{X}]{\A{size}[][*]}{\A{ribbon}[][*]}\
does not hold with respect to some context \set{X}.
The latter is a weaker statement.

\begin{example} \label{ex:GPMonths}
	Consider Table~\ref{table:festivals}
	and the \IIOC\ of
	\oc[\brac{\A{country}}]{\A{size}[][*]}{\A{ribbon}[][*]}.
	Figures~\ref{fig:size-ribbon-china} and \ref{fig:size-ribbon-canada} show
	the two \BG s
	for \emph{China} and \emph{Canada}
	(the values of our context, \A{country}),
	respectively.
	Thus,
	there exists
	a co-order between \A{size}\ and \A{ribbon}\
	over \(\set{E}({\tup{t}[1]}_{\A{country}})\)
	(that is, for \(\A{country} = \mbox{\emph{`China'}}\))
	and a co-order between \A{size}\ and \A{ribbon}\
	over \(\set{E}({\tup{t}[6]}_{\A{country}})\)
	(that is, for \(\A{country} = \mbox{\emph{`Canada'}}\)).
\end{example}

We next need to show how to extract a co-order
once we know one exists.
As with Sections~\ref{sec:eioc-empty-context}
    and
    \ref{section:nonemptyEI},
we may discover a partial order,
this time both for \emph{left} and \emph{right},
within each partition group with respect to the context. 
The partial order is of a specific type:
we find a disjoint collection of \emph{chains}.
Each chain is a strong total order over its values. Note that the singleton elements (which were initially ignored in $\BGp$) will be inserted into this order, creating the final order.
Again,
there is no specified direction
in which to read each chain; i.e., 
what its polarity is.


If, for each partition group with respect to \set{X}\ over \T{r},
\BGp[\A{A}, \A{B}]\ over the partition group
satisfies the conditions
in Theorem~\ref{theo:bipValidity},
then the conditional \IIOC\ of
\oc[\set{X}]{\A{A}[][*]}{\A{B}[][*]}\
holds over \T{r}.
\BGp[\A{A}, \A{B}]\ over each partition group
yields a strong partial order%
---%
a disjoint collection of chains%
---%
for each of \A{A}\ and \A{B}.
A walk of
\BGp[\A{A}, \A{B}]\ 
suffices to enumerate the chains,
for both \A{A}\ and \A{B}.

\begin{example}
Consider the \BG\
in Figure~\ref{fig:size-ribbon-canada}
over the \IIOC\ of
\oc[\brac{\A{country}}]{\A{size}[][*]}{\A{ribbon}[][*]},
over values
in Table~\ref{table:festivals}.
By iteratively zig-zagging from left to right in this bipartite graph,
we obtain the chains
\Lst{\emph{Small}, \emph{Medium}, \emph{Large}}\
and
\Lst{\emph{White}, \emph{Blue}, \emph{Red}}\
over \A{size}\ and \A{ribbon}, respectively,
over 
partition group
\(\set{E}({\tup{t}[6]}_{\A{country}})\).
\end{example}





As in Section~\ref{section:nonemptyEI},
an \IIOC\ with a non-empty context can be treated
either as \emph{conditional} or \emph{unconditional}.
Our discovered domain orders between partition groups
with respect to the context
may differ.
For the conditional case,
this is considered fine; e.g.,
in Table~\ref{table:festivals},
the order of ribbon colors
w.\ r.\ t.\ the festival size
differs per \A{country}:
in \emph{China},
	\(
		\emph{White}
			\prec
		\emph{Red}
			\prec
		\emph{Blue}
	\);
in \emph{Canada},
	\(
		\emph{White}
			\prec
		\emph{Blue}
			\prec
		\emph{Red}
	\).



\begin{example}
In Table~\ref{table:festivals},
the conditional \IIOC\ of
\oc[\brac{\A{country}}]{\A{size}[][*]}{\A{ribbon}[][*]}\
holds as 
\(
	\set{E}({\tup{t}[1]}_{\A{country}})
		\models
	\oc{\A{size}[][*]}{\A{ribbon}[][*]}
\)
and
\(
	\set{E}({\tup{t}[6]}_{\A{country}})
		\models
	\oc{\A{size}[][*]}{\A{ribbon}[][*]}
\).
\end{example}

\begin{lemma}
\blue{The runtime of validating a conditional $\IIOC$ with empty or non-empty context is ${\mathcal O}(n)$.}
\end{lemma}




\subsection{I/I OCs, Unconditional} \label{ii-with-context}

\blue{
To validate an $\IIOC$ candidate with a non-empty context \emph{unconditionally}
and find implicit orders $\A{A}^*$ and $\A{B}^*$ that hold over the \emph{entire} dataset
is significantly harder.} 
The implicit orders for left and for right
discovered per partition group
must be consistent and
polarity choices must be made for them.

For example,
the months in the Gregorian and lunar calendars
are dependent in the context of the year types
with respect to
the \IIOC\ of
\oc[\brac{\A{yearGreg}, \A{yearLun}}]
   {\A{monthGreg}[][*]}
   {\A{monthLun}[][*]}.
In the lunar calendar,
there are twelve months (sometimes, thirteen),
with the new year starting a bit later
than in the Gregorian calendar,
with the lunar months overlapping the Gregorian months.

We prove that this is computationally hard.
To do this,
we show that to determine
whether,
for a collection of chains,
there exists a polarization,
a directional choice for each chain,
such that the union of the chains so directed
represent a strong partial order.
This is a sub-problem for deciding the validity of an \IIOC;
therefore,
this establishes that our problem is hard.

\begin{definition}\label{def:cpp}
(\emph{the chain polarity problem})
\label{DEF/CPP}
For the \emph{Chain Polarity Problem (CPP)},
the \emph{structure}
is a collection of lists of elements. Each list is constrained such
that no element appears twice in the list. A list can be interpreted
as defining a total order over its elements;
e.g., list \Lst{\var{a}, \var{b}, \var{c}, \var{d}}\ infers
$a \prec b$, $b \prec c$, and $c \prec d$.

\begin{table}[t]
\paddingT
\centering
    \caption{\label{tab:nae-to-cpp} A NAE-3SAT instance and the reduced CPP instance.}
    \paddingT
    \renewcommand{\arraystretch}{1.1}
    \scalebox{0.87}{
\begin{tabular}{|c|c|}
\hline
Clauses & Lists \\ \hline
$(p_1\lor p_2\lor \neg p_3)$ & $[t_1, a_1, b_1, f_1]$, $[t_2, b_1, c_1, f_2]$, $[f_3, c_1, a_1, t_3]$ \\ \hline
$(\neg p_1\lor p_2\lor \neg p_3)$ & $[f_1, a_2, b_2, t_1]$, $[t_2, b_2, c_2, f_2]$, $[f_3, c_2, a_2, t_3]$ \\ \hline
$(\neg p_1\lor\neg p_2\lor p_3)$ & $[f_1, a_3, b_3, t_1]$, $[f_2, b_3, c_3, t_2]$, $[t_3, c_3, a_3, f_3]$ \\ \hline
\end{tabular}
}
\end{table}

A \emph{polarization} of the list collection is a new list collection
in which, for each list in the original, the list or the
reverse appears.
The decision question 
for CPP is whether there
exists a polarization of the CPP instance such that the union of
the total orders represented by the polarization’s lists is a strong
partial order. 

\end{definition}


\begin{lemma}\label{lemma:NPcomplete}
The Chain Polarization Problem is NP-Complete.
\end{lemma}

We prove Lemma~\ref{lemma:NPcomplete} using a reduction from NAE-3SAT, which is a variation of 3SAT that requires that the three literals in each clause are not all equal to each other.

\eatproofs{
\begin{proof}
The input size of a CPP instance may be measured as the sum of the lengths of its lists; let this be $n$. Consider a pair explicitly implied by the list collection to be in the binary ordering relation if the pair of elements appears immediately adjacent in one of the lists. Thus, the number of explicitly implied pairs is bounded by $n$.

\textbf{CPP is in the class NP.}
An answer of yes to the corresponding decision question means there exists a polarization of the CPP instance that admits a strong partial order. Given such a polarization witness, its validity can be checked in polynomial time. The size of the polarization is at most $n$. The set of explicitly implied ordered-relation pairs is at most $n$. Computing the transitive closure over this set of pairs is then polynomial in $n$. If no reflexive pair (e.g., a $\prec$ a) is discovered, then there are no cycles in the transitive closure, and thus this represents a strong partial order. Otherwise, not.

\textbf{CPP is NP-complete.}
The known NP-complete problem NAE-3SAT (Not-All-Equal 3SAT) can be reduced to CPP.

The structure of a NAE-3SAT instance is a collection of clauses. Each clause consists of three literals. A literal is a propositional variable or the negation thereof. A clause is interpreted as the disjunction of its literals, and the overall instance is interpreted as the logical formula which is the conjunction of its clauses. Since each clause is of fixed size, the size of the NAE-3SAT instance may be measured by its number of clauses; call this $n$.

The decision question for NAE-3SAT is whether there exists a truth assignment to the propositional variables (propositions) that satisfies the instance formula such that, for each clause, at least one of its literals is false in the truth assignment and at least one is true. 

We can establish a mapping from NAE-3SAT instances to CPP instances which is polynomial time to compute and for which the decision questions are synonymous.
Let $p_{1}, .., p_{m}$ be the propositions of the NAE-3SAT instance. For clause \emph{i}, let ($L_{i,1}, L_{i,2}, L_{i,3}$) represent it, where each $L_{i,j}$ is a placeholder representing the corresponding proposition or negated proposition as according to the clause.
We build a corresponding CPP instance as follows. For each clause, we add three lists. For clause \emph{i}, add
[$X_{i,1}, a_i, b_i, Y_{i,1}$],
[$X_{i,2}, b_i, c_i, Y_{i,2}$], and
[$X_{i,3}, c_i, a_i, Y_{i,3}$].
Each $X_{i,j}$ and $Y_{i,j}$ are placeholders above, and correspond to the $L_{i,j}$’s in the clauses.

We replace them in the lists as follows. There are two cases for each: $L_{i,j}$ corresponds to a proposition or the negation thereof. Without loss of generality, let $L_{i,j}$ correspond to $p_k$ or to $\neg p_k$. If $L_{i,j}$ corresponds to $p_k$, we replace $X_{i,j}$ with element $t_k$ and $Y_{i,j}$ with element $f_k$. Else ($L_{i,j}$ corresponds to $\neg p_k$), we replace $X_{i,j}$ with $f_k$ and $Y_{i,j}$ with $t_k$.
Call the $t_i$ and $f_i$ elements propositional elements, and call the $a_i$, $b_i$, and $c_i$ elements confounder elements.

There exists a polarization of the CPP instance that admits a strong partial order iff the corresponding NAE-3SAT instance is satisfiable such that, for each clause, at least one of its literals has been assigned false.
For the propositional elements, let us interpret $t_i \prec f_i$ in the partial order as assigning proposition $p_i$ as true; and $f_i \prec t_i$ as assigning it false.

For any polarization of the CPP instance that admits a strong partial order, for each clause, $i$, one or two, but not all three, of the corresponding lists must have been reversed. Otherwise, there will be a cycle in the ordering relation over the confounder elements, $a_i$, $b_i$, and $c_i$.
Since not all three lists for clause $i$ can be reversed, $X_{i,1} \prec Y_{i,1}, X_{i,2} \prec Y_{i,2}$, or $X_{i,3} \prec Y_{i,3}$. And thus, at least one of the clause’s literals is marked as true, and so the clause is true. Since at least one of the three lists for clause i must be reversed, we know that $Y_{i,1} \prec X_{i,1}, Y_{i,2} \prec X_{i,2}$, or $Y_{i,3} \prec X_{i,3}$. And thus, not all of the clause’s literals are marked as true.
If two lists (coming from different clauses) “contradict” — that is, one implies, say, $t_i \prec f_i$ and the other that $f_i \prec t_i$ — then there would be a cycle in the transitive closure of the ordering relation. This would contradict that our polarization admits a strong partial order. Thus, for each $p_i$, the partial order either has $t_i \prec f_i$ or $f_i \prec t_i$. This corresponds to a truth assignment that satisfies the NAE-3SAT instance.

In the other direction, if no polarization of the CPP instance admits a strong partial order, then there is no truth assignment that satisfies the NAE-3SAT instance such that, for each clause, at least one of its literals has been assigned false.
\end{proof}
}

Thus, our problem being 
NP-hard follows
from Lemma~\ref{lemma:NPcomplete}
(as it is a superset of the CPP problem),
and 
as
a solution
that an unconditional \IIOC\ is valid
can be verified in polynomial time,
it is NP-complete. 

\begin{theorem}\label{theo:nphardness}
The problem of validating a given unconditional $\IIOC$ with non-empty context is NP-complete.
\end{theorem}


\begin{example}
Table~\ref{tab:nae-to-cpp} illustrates a NAE-3SAT instance with three clauses and its equivalent CPP instance with nine lists. In the CPP instance, $t_i \prec f_i$ in the partial order is interpreted as assigning proposition $p_i$ as \emph{true}, and $f_i \prec t_i$ as assigning it \emph{false}.
Also, the variables $a_i$, $b_i$ and $c_i$ ensure that there exists at least one true and one false assignment for the literals in each clause. This condition is satisfied as among the three lists generated for each clause, exactly one or two of them have to be reversed in order to avoid a cycle among $a_i$, $b_i$ and $c_i$. This translates to the corresponding literals having false assignment and the rest true assignments. Hence, any valid polarization for the lists in the CPP instance can be translated to a valid solution for the NAE-3SAT instance. 
\end{example}

We illustrate our approach
for validating \IIOC\ candidates
by employing a high quality SAT solver
in Section~\ref{SEC:SAT-Solver}.



\subsection{I/I ODs} \label{SEC/IIODs}

Discovery of domain orders via \IIOD s
with an \emph{empty context}
(or with a non-empty context but considered \emph{conditionally})
is essentially impossible.
While we can discover \IIOD s that hold over the data,
we can only 
infer the \emph{empty} order for the domains.
The \FD\ essentially masks any information that could be derived
about the 
orders.

\begin{theorem}
If \fd{\set{X}\A{A}}{\A{B}},
then the conditional \IIOD\ candidate
\oc[\set{X}]{\A{A}[][*]}{\A{B}[][*]}\ must be valid.
Furthermore,
there is a unique partial order
that can be derived
for \A{A}[][*]\ and for \A{B}[][*]:
the empty order.
\end{theorem}

\begin{table}
\paddingT
\parbox{.45\linewidth}{
\centering
    \caption{\label{tab:iiod-sample-database} Valid ${\sf I/I\ OD}$.}
    \paddingT
    \scalebox{0.87}{
    \begin{tabular}{c |c c c}
\# & \A{year} & \A{month} & \A{version\#} \\ \hline 
\tup{t}[1] & 2018 & Jan & v99 \\ 
\tup{t}[2] & 2018 & Feb & v100 \\ 
\tup{t}[3] & 2019 & Jan & v99 \\ 
\tup{t}[4] & 2019 & March & v100 \\ 
\tup{t}[5] & 2020 & Feb & v99 \\ 
\tup{t}[6] & 2020 & March & v100 \\ 
\end{tabular}
}
}
\hfill
\parbox{.45\linewidth}{
\centering
    \caption{\label{tab:cpp-sat-reduction-sample} CPP to SAT.}
    \paddingT
     \scalebox{0.87}{
    \begin{tabular}{ c| c c c }
\# & \A{C} & \A{A} & \A{B} \\ \hline 
\tup{t}[1] & $1$ & $1$ & $1$ \\ 
\tup{t}[2] & $1$ & $2$ & $2$ \\ 
\tup{t}[3] & $1$ & $3$ & $3$ \\ 
\tup{t}[4] & $2$ & $1$ & $4$ \\ 
\tup{t}[5] & $2$ & $2$ & $5$ \\ 
\tup{t}[6] & $2$ & $4$ & $5$ \\ 
\tup{t}[7] & $2$ & $4$ & $6$ \\ 
\end{tabular}
}
}
\paddingD
\end{table}

\begin{example}
Consider the \IIOC\ 
\oc{\A{festival}^*}{\A{monthGreg}^*}\
and Table~\ref{table:festivals}.
Since the \FD\ of \fd{\A{festival}}{\A{monthGreg}}\ holds,
the empty partial order is the 
implicit order
over 
\A{monthGreg}.
\end{example}

However,
in the case of a candidate \IIOC\
with a non-empty context considered unconditionally
paired with an \FD\
that holds also ``only within a non-empty context'',
it \emph{is} possible for us to discover meaningful
domain orders.

\begin{example}
Consider Table~\ref{tab:iiod-sample-database},
which shows different versions
of a software released in each year and month,
and the unconditional \IIOD\
of \od[\brac{\A{year}}]{\A{month}}{\A{version\#}^*}.
The only valid strong partial orders
over the values of \A{month}\
and 
\A{version\#}\
are
    \(\emph{Jan} \prec \emph{Feb} \prec \emph{March}\)
        and
    \(\emph{v99} \prec \emph{v100}\),
or the reversals of these,
respectively.
\end{example}

\section{Using a SAT Solver for I/I OC's}
\label{SEC:SAT-Solver}


Given that discovering implicit domain orders via \IIOC s is NP-hard,
we reduce it to an instance of the SAT problem
to validate the candidate
and then to establish
valid strong partial orders.
The first step is similar to the conditional case
in Section~\ref{sec:iioc-empty-context}:
we derive bipartite graphs, \BG[i]'s,
for the tuples from each partition group.
Presence of cyclicity or 3-fan-out invalidates the candidate,
as by Theorem~\ref{theo:bipValidity}.
Thus next,
we check each \BG[i]\ for cyclicity or 3-fan-out;
this validates or invalidates the candidate in linear time.%
\footnote{
    The constraint that each \BG[i]\ has no 3-fan-out
    restricts the size of the graph to be linear
    in the number of distinct values of the domain.
    Without this,
    the size of the graph could be quadratic
    in the number of distinct values of the domain.
}

To translate an instance of our \IIOC\ validation problem
isomorphically into a SAT instance,
we create two types of clauses:
one type to capture the constraints on order
implied by the data;
and the second type
that encodes transitive closure
over these.



\begin{enumerate}[nolistsep,leftmargin=*]
\item \textbf{No Swaps.} 
    \blue{To enforce an order between pairs of values in each partition group without causing a \emph{swap}.}
    
    
\item \textbf{Transitivity.}
    To guarantee a valid strong partial order,
    we need transitive closure to check for cycles.
\end{enumerate}

We now explain the translation to SAT.
The input is 
a list of bipartite graphs, \BG's,
which indicate the co-occurrence of values
in each partition.
For each 
bipartite graph \BG[j],
the left-hand side (LHS) and right-hand side (RHS) denote
the values of the two attributes of the \IIOC\ candidate,
respectively.
Let $l_1, \dots, l_m$ and $r_1, \dots, r_k$
denote the distinct values of each.

We define two sets of propositional variables:
$\forall 1\le i, j\le m: l_{i, j}$ and $\forall 1\le i, j\le k: r_{i, j}$. Assigning \emph{true} to a variable $l_{i, j}$ indicates $l_i\prec l_j$, while assigning \emph{false} means that either $l_j\prec l_i$, or the order between these values has not been discovered. Thus, for every two variables $l_{i,j}$ and $l_{j,i}$,  $\neg(l_{i,j}\land l_{j,i})\equiv(\neg l_{i,j}\lor\neg l_{j,i})$ is added, as these variables cannot both be \emph{true}. The same applies for variables $r_{i,j}$.

\textbf{No swaps.} 
\blue{$\forall (l_u, r_t),$ $(l_v, r_w)\in {\sf BG_i}, l_u\ne l_v\text{ and }r_t\ne r_w: (l_{u,v}\land r_{t,w})\lor(l_{v,u}\land r_{w, t})\equiv(l_{u,v}\lor r_{w,t})\land(l_{v,u}\lor r_{t,w})$. 
Note that the initial
conditions ($(\neg l_{u,v}\lor\neg l_{v,u})$ and $(\neg r_{t,w}\lor\neg r_{w,t})$) were used to simplify these conditions.}

\textbf{Transitivity.} Next, we add the following clauses to ensure that transitivity is satisfied: $\forall 1\le u, v, w\le m\land u, v, w$ \mbox{distinct}: $(l_{u, v}\land l_{v, w})\implies l_{u, w}$. We add similar clauses for the RHS values. 

\begin{theorem}
The unconditional $\IIOC$ candidate is valid \emph{iff} the corresponding SAT instance is satisfiable.
\end{theorem}

If the SAT instance is satisfiable, to derive the final partial orders
over the values of $\A{A}$ and $\A{B}$, we take the satisfying assignment 
and set $i\prec j$ for $\A{A}$ iff $l_{i,j}=\emph{true}$, and similarly 
for the values of $\A{B}$.
\blue{
To achieve a pair of \emph{strongest derivable orders}, 
we remove the order over pairs of values which should not exist in the final order, 
while keeping the order graph valid.
For every pair of distinct values $l_u, l_v\in\PG_i$ where $l_u\prec l_v$, 
we keep them in the final order \emph{iff} one of these conditions holds 
(and similarly for the RHS values):
1) the nodes $l_u$ and $l_v$ 
are in the same connected component in $\sf{BG}_i$ 
and the path from $l_u$ to $l_v$ contains at least two nodes with degree two or larger;
2) $\exists\sf{BG}_j, j\ne i$ and distinct values $v_r, v_s$ belonging to 
the same attribute, $s.t.$ $v_r,$ $v_s$ $\in$ $\BG_j$ $\land$ $v_r$ and $v_s$ 
are in the same connected component in $\BG_i$ as $l_u$ and $l_v$, 
respectively
(note that $v_r$ and $v_s$ could be the same as $l_u$ and $l_v$).  
Intuitively, either of these conditions would make it impossible to \emph{remove} an order between two values through valid \emph{transposition}s within the same \emph{witness class}, as defined in Section~\ref{sec:pairs of implicit domains}.
}
\diagram{These steps for validating an unconditional $\IIOC$ 
candidate are included in Figure~\ref{fig:diagram-ii}.}

\begin{example}
Consider Table~\ref{tab:cpp-sat-reduction-sample} and the $\IIOC$ $\simularC{\A{C}}{\A{A}^*}{\A{B}^*}$ candidate. The propositional variables for the reduction are $l_{1,2},$ $l_{2,1},$ $\dots,$ $l_{4,3}$ and $r_{1,2},$ $\dots,$ $r_{6,5}$.
Initial clauses are generated for each pair of values to ensure that an order in both directions is not established, e.g., $(\neg l_{1,2}\lor\neg l_{2,1})\land(\neg l_{2,3}\lor\neg l_{3,2})$.
For the \emph{no swaps} condition, the clauses $(l_{1,2}\lor r_{2,1})\land(l_{2,1}\lor r_{1,2})$ are generated for tuples $\tup{t}_1$ and $\tup{t}_2$, and similarly for the rest of the pairs of tuples in $\BG_1$ and $\BG_2$.
%
%
%
Finally, we add the clauses $(l_{1,2}\land l_{2,3})\implies l_{1,3}$, $(l_{1, 3}\land l_{3,2})\implies l_{1,2}$, and accordingly the remaining clauses to capture the \emph{transitivity}. 

\blue{
Since this is a valid $\IIOC$, some strong partial order can be derived 
using the SAT variable assignments. To derive the \emph{strongest} orders 
from this pair, the final orders between values such as (2, 4) and (1, 2)
on the LHS are kept, as these pairs satisfy the first and 
second condition, respectively (values 2 and 4 exist in a 
connected component in $\BG_2$ with two nodes with degree two along their path, 
and values 1 and 2 are present on the LHS of both $\BG_1$ and $\BG_2$). 
However, the relation between values 1 and 3 (and similarly 2 and 3) 
is removed since these values do not satisfy any of the conditions. 
This process is repeated for the RHS as well.
}
\end{example}

\begin{lemma}
\blue{The cost of the SAT reduction is ${\mathcal O}(n+pm^2+m^3)$.}
\end{lemma}

Since $m$ tends to be small in practice (i.e., $m\ll |\R{r}|$) for meaningful cases 
and $p$ heavily depends on $m$ as well, this runtime is manageable in 
real-life applications (Section~\ref{sec:experiments}).

We also use an optimization based on the overlap between values in 
different partition groups to decrease the number of variables and clauses. 
Initially, we compute disjoint sets of values that have co-appeared 
in the same partition groups. While considering pairs or 
triples of values to generate variables or transitivity clauses, 
respectively, 
we consider these sets of attributes separately, 
which reduces the runtime of the algorithm, 
without affecting the correctness of the reduction. 
Furthermore, before reducing to a SAT instance, 
we consider pairwise $\BG$s and check their compatibility, 
in order to falsify impossible cases as early as possible. 
\section{Measure of Interestingness}
\label{sec:measure-interestingness-def}

The search space and the number of discovered implicit domain orders
may be large in practice. 
\blue{Inspired by previous work \cite{SGG+2017:ODD,chiang-data-quality-2008},}
to decrease the cognitive burden of human verification, we propose a measure of \emph{interestingness} to rank the discovered domain orders based on how 
close each is to being
a strong total order. 
\blue{We argue that by focusing on similarity to a strong total order, this measure is successful in detecting meaningful and accurate implicit orders, and 
demonstrate this in our experiments in Section~\ref{sec:experiments}.
}

Given a DAG $G$ representing a strong partial order, the \emph{pairwise} interestingness measure is defined as \emph{pairwise}$(G)=|\emph{pairs}(G)| / \binom{m}{2}$, where $\emph{pairs}(G)$ = $\{(u, v): u, v \in G$ \emph{and there is a path between} \text{$v$  \emph{and} $u$}\}, and $m$ is the number of vertices in $G$. The number of pairs of vertices that are connected demonstrates the quality of the found strong partial orders, while the binomial coefficient in the denominator is used to normalize. 
Based on this measure, a strong total order graph has the perfect score of $1$, while a completely disconnected graph has a score of $0$.

\begin{example}
Consider the order graph $G$ presented in Figure~\ref{fig:eioc-num-lunar-merged}. There are 23 pairs of connected vertices 
and $\binom{8}{2}=28$ possible pairs. Thus, the pairwise score is $\emph{pairwise}(G)$ $=$ $\frac{23}{28}$ $\approx$ $0.82$.
\end{example}

For conditional implicit orders, we divide the number of \emph{pairs} 
in each partition group over the total number of pairs possible among \emph{all} 
the values in the attribute, and then compute their average. This is to prevent 
candidates with many partition groups with less interesting partial orders from 
achieving a high score. To achieve a score of 1, the partial order in each
partition group needs to be strong total order over \emph{all} the values in the attribute. 
%
Our algorithm for computing this measure may take quadratic time ${\mathcal O}(m^2)$ in the number of vertices in the graph, which corresponds to the number of unique elements in the attribute.
This is not significant, in practice (Section~\ref{sec:experiments}). 





\newcounter{exp-c}
\setcounter{exp-c}{0}

\section{Experiments} \label{sec:experiments}


We implemented our implicit domain order discovery algorithm, named $\iOrder$,
on top of a Java implementation of the  set-based $\EEOD$ discovery algorithm~\cite{SGG+2017:ODD,SGG2018:BiODD}. 
Furthermore, we use Sat4j, 
which is an efficient SAT solver library~\cite{le2010sat4j}. 
Our experiments were run on a machine with a Xeon CPU 
2.4GHz with 64GB RAM. 
%
We use two integrated datasets from the 
Bureau of Transportation Statistics (BTS) and the North Carolina Board of Elections (NCSBE):\blue{\footnote{\blue{These datasets can be accessed through \url{https://cs.uwaterloo.ca/~mkaregar/datasets/}.}}}

\begin{itemize}[nolistsep,leftmargin=*]
    \item \textbf{Flight} 
    contains information about flights in the US with 1M tuples and 35 attributes (\url{https://www.bts.gov}). 
    \item \textbf{Voter} 
    contains data  about voters in the US with 1M tuples and 35 attributes
    (\url{https://www.ncsbe.gov}).

\end{itemize}
\blue{We chose these datasets due to their size for scalability experiments and for having real-life attributes with interesting implicit orders.}
\subsection{Scalability}
\label{sec:exp:scal}

\refstepcounter{exp-c}\label{exp-scal-tup}
\textbf{Exp-1: Scalability in $|\T{r}|$.}
We measure the running time of $\iOrder$ by varying the number of tuples 
(Figure~\ref{figure:scalTup}). We use the $\sf{Flight}$ and $\sf{Voter}$ datasets
\blue{with 10 attributes} and up to 1M tuples, \blue{by showing data samples to users and asking them to mark attributes as potential candidates.}
\blue{In the absence of user annotation, the algorithm can be run multiple times
over different subsets of attributes 
to capture potential implicit orders.}
Figure~\ref{figure:scalTup} shows a linear runtime growth as computation is dominated by the verification of $\OD$ candidates, and the non-linear factors in our algorithm
tend to have less impact. Thus, $\iOrder$ scales well for large datasets. 


\refstepcounter{exp-c}\label{exp-scal-attr}
\textbf{Exp-2: Scalability in $|\T{R}|$.}
Next, we 
vary the number of attributes. We use the $\sf{Flight}$ and $\sf{Voter}$ datasets with 1K tuples 
(to allow experiments with a large number of attributes in reasonable time)  
and up to 35 attributes. Figure~\ref{figure:scalAttr} illustrates that the running time increases exponentially with the number of attributes (the Y-axis is in log scale). This is because the number of implicit order candidates is exponential in the worst case.
The $\sf{Voter}$ dataset requires more time for the same number of attributes due to a larger number of candidates.

\refstepcounter{exp-c}\label{exp-npc-cases}
\textbf{Exp-3: NP-complete Cases in Practice.} 
The most general case of implicit domain order discovery through unconditional $\IIOC$s is NP-complete (Section~\ref{ii-with-context}).
However, the majority of observed cases took a short time. 
The cases reduced to SAT were on average solved in under 60 ms in Exp-\ref{exp-scal-tup} and Exp-\ref{exp-scal-attr}, indicating that NP-complete cases are handled well in practice. 
In Exp-\ref{exp-scal-attr}, with varying the number of attributes and 1K tuples, on average, 
33\% of the total runtime was spent on reducing to, and solving, the SAT instances.  However, in the corresponding Exp-\ref{exp-scal-tup},  with varying the number of tuples up to 1M tuples, this ratio was less than 1\%. This is attributed to the large number of 
candidates in Exp-\ref{exp-scal-attr} 
and the small number of tuples, so the linear steps of the algorithm do not dominate the runtime.

\begin{figure}[t]
    \centering
    \begin{subfigure}[]{0.23\textwidth} 
        \centering
        \includegraphics[width=4.05cm]{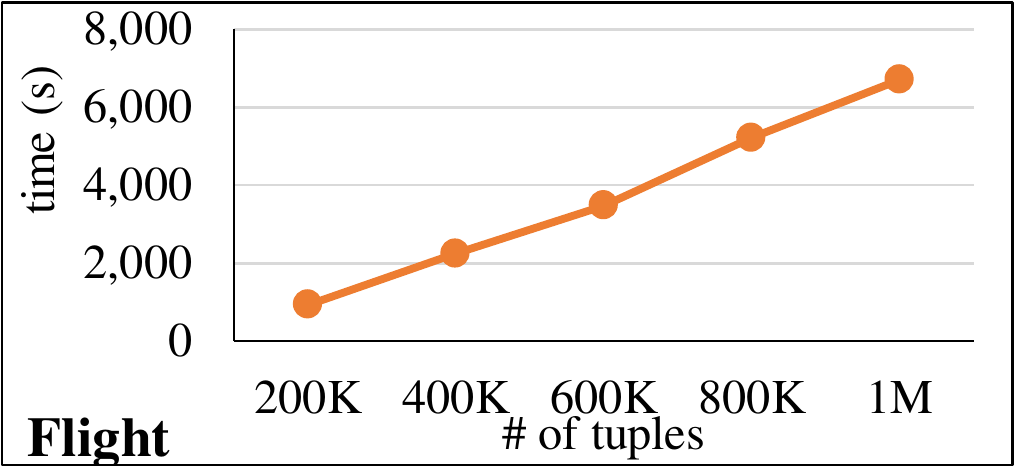}
    \end{subfigure}%
    ~ 
    \begin{subfigure}[]{0.23\textwidth} 
        \centering
        \includegraphics[width=4.05cm]{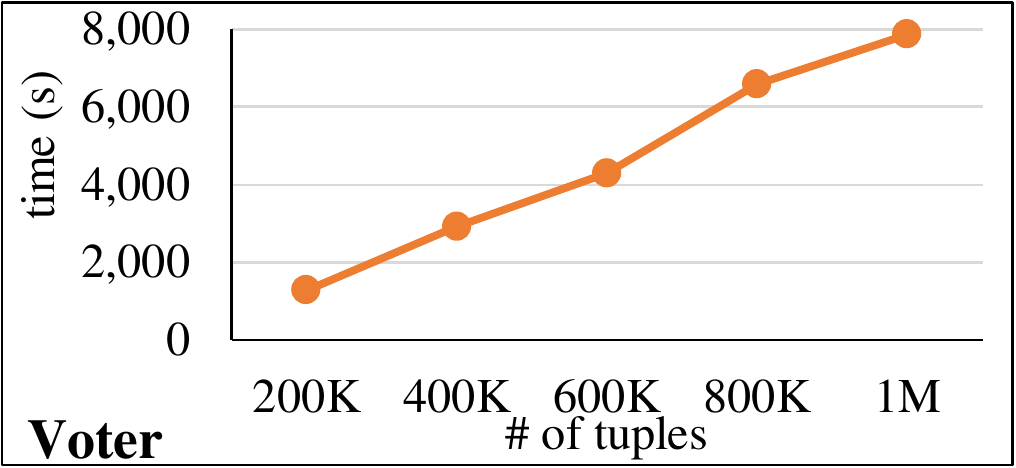}
    \end{subfigure}
    \paddingT
    \caption{Scalability and effectiveness in $|\textbf{r}|$.}
    \paddingD
    \label{figure:scalTup}
\end{figure}

\begin{figure}[t]
    \centering
    \begin{subfigure}[]{0.23\textwidth} 
        \centering
        \includegraphics[width=4.05cm]{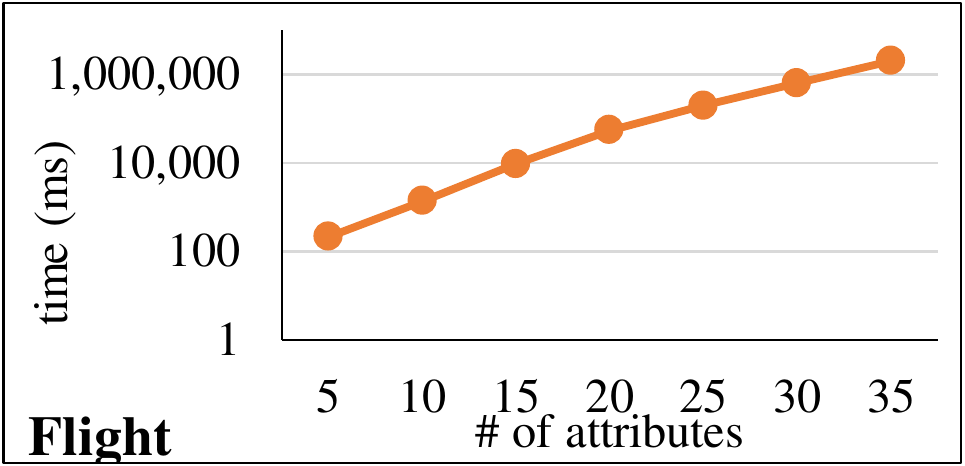}
    \end{subfigure}%
    ~ 
    \begin{subfigure}[]{0.23\textwidth} 
        \centering
        \includegraphics[width=4.05cm]{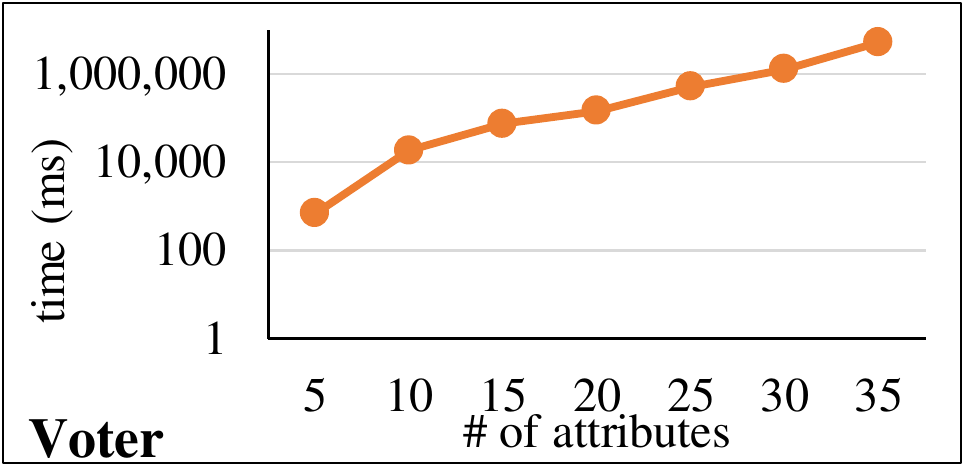}
    \end{subfigure}
    \paddingT
    \caption{Scalability and effectiveness in $|\textbf{R}|$.}
    \paddingD
    \label{figure:scalAttr}
\end{figure}

\refstepcounter{exp-c}\label{exp-lattic-level}
\subsection{\blue{Effectiveness}}
\label{sec:exp:lattice}

\textbf{Exp-4: Effectiveness over lattice levels.}
Here, we measure the running time and the number of discovered implicit domain orders at different levels of the lattice (Figure~\ref{figure:lattice-runtime}). We report the results with 10 attributes over 1M tuples (from Exp-\ref{exp-scal-tup}) 
in the $\sf{Flight}$ and $\sf{Voter}$ datasets. Since the attribute lattice 
is diamond-shaped and nodes are pruned over time through axioms, the time to process each level first increases, up to level five, and decreases thereafter.

As most of the interesting implicit orders are found at the top levels 
with respect to a smaller context (as verified in Exp-\ref{exp-interestingness}), 
we can prune the lower levels 
to reduce the total time. 
In the $\sf{Flight}$ and $\sf{Voter}$ datasets, approximately 85\% and 73\% of the orders are found in the first three levels, taking about 45\% and 19\% of the total time, respectively.
\blue{In the $\sf{Voter}$ dataset, fewer implicit orders are found in the first levels of the lattice, creating fewer pruning opportunities. Therefore, more time is spent on validating candidates with larger contexts, 
explaining the runtime difference between the two datasets.}



\refstepcounter{exp-c}\label{exp-interestingness}
\textbf{Exp-5: Interestingness of implicit orders.}
We argue that implicit domain orders found at upper levels of the lattice are the most interesting. Implicit orders found with respect to a context with more attributes contain more partition groups. 
Hence, an implicit order with respect to a less compact context may hold, but may not be as meaningful, due to overfitting. Figure~\ref{figure:lattice-runtime} illustrates that the interestingness score 
drops from the fourth level on for the $\sf{Flight}$ dataset and from the third level on  for the $\sf{Voter}$ dataset. 

Figure~\ref{fig:interestingness-scores} illustrates that our interestingness measure can reduce the 
\blue{number of} 
implicit domain orders. 
\blue{In the $\sf{Flight}$ dataset, other than $\A{monthGreg}^*$ and $\A{monthLunar}^*$, we found a high-scoring order   $\A{delayDesc}^*$, which orders flight delay as
\emph{Early} $\prec$ \emph{On-time} $\prec$ 
\emph{Short delay} $\prec$ \emph{Long delay},  
and is discovered through the $\EIOD$ $\simular{\A{delay}}{\A{delayDesc}^*}$.
This order is interesting since long delays may result in fines on the airline, so detecting these instances is valuable.  Implicit orders over  $\A{distanceRangeMile}$ and 
$\A{distanceRangeKM}$ were discovered through the $\IIOC$
$\simular{\A{distanceRangeMile}^*}{\A{distanceRangeKM}^*}$, 
as the categories for these attributes are overlapping 
(e.g., ranges $0-700$ and $700-3000$ miles overlap with the range of $1100-4800$ kilometers). 
Another ordered attribute in this dataset is $\A{flightLength}$
(\emph{Short-haul} $\prec$ \emph{Medium-haul} $\prec$ \emph{Long-haul}),  
which was discovered through the $\EIOD$
$\simularC{\A{airline}}{\A{flightDuration}}{\A{flightLength}^*}$. 
The non-empty context is due to different airlines using 
different ranges to define flight duration.
In the $\sf{Voter}$ dataset, the attributes $\A{ageRange}$ 
($12-17$ $\prec$ $18-24$ $\prec$ $\dots$ $\prec$ $+75$) 
and $\A{generation}$ 
(e.g., \emph{Baby Boomer} $\prec$ \emph{Generation X} 
$\prec$ \emph{Millennial} $\prec$ \emph{Generation Z}) 
are discovered through $\EIOD$s $\simular{\A{age}}{\A{ageRange}^*}$
and $\simular{\A{birthYear}}{\A{generation}^*}$,  respectively, as well as the $\IIOC$
$\simular{\A{ageRange}^*}{\A{generation}^*}$
as the categories are overlapping. 
Finally, an order over $\A{birthYearAbbr}$
was detected using the \emph{conditional} $\EIOD$ 
$\simularC{\A{isCentenarian}}{\A{birthYear}}{\A{birthYearAbbr}^*}$. 
This is because for the people born in the early 1917 or before with $\A{isCentenarian}=True$
we have $'97$ $\prec$ $'98$ $\prec$ $\dots$ $\prec$ $'16$ $\prec$ $'17$, 
while for people born after this date with $\A{isCentenarian}=False$ we have 
$'17$ $\prec$ $'18$ $\prec$ $\dots$ $\prec$ $'98$ $\prec$ $'99$. Thus an unconditional $OD$ is not possible. 
}


\begin{figure}[t]
  \paddingT
    \center
    \scalebox{0.95}{
    \includegraphics[width=0.45\textwidth]{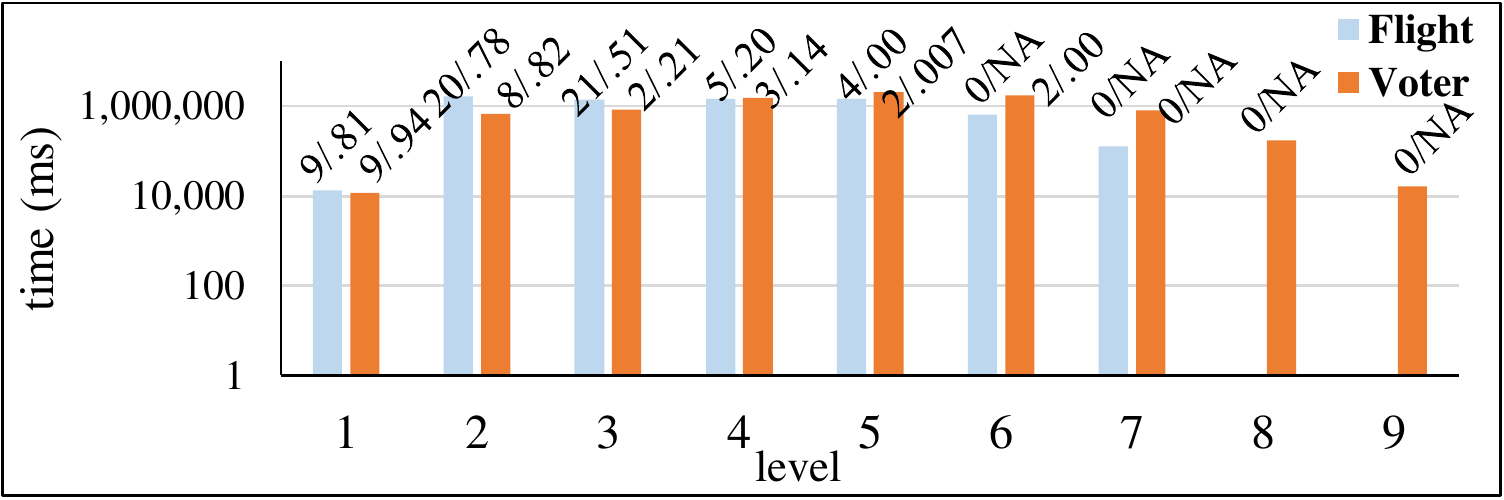}
    }
        \paddingT
     \caption{Runtime, \# of orders, and average interestingness.}
     \label{figure:lattice-runtime}
        \paddingD
\end{figure}

\begin{figure}[t]
    \center
    \scalebox{0.95}{
    \includegraphics[width=0.45\textwidth]{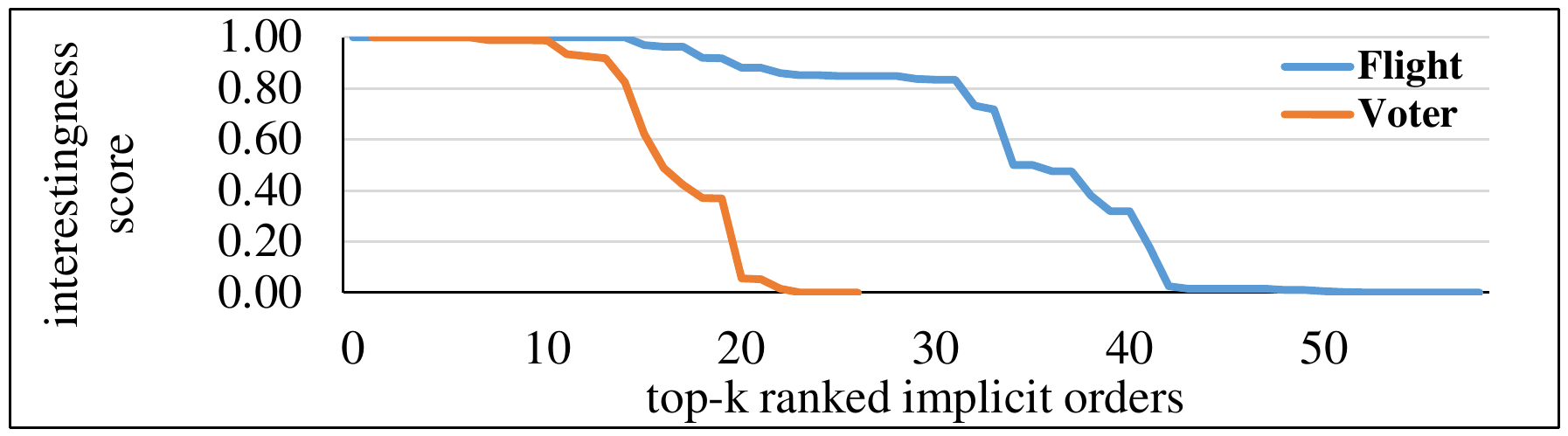}
    }
        \paddingT
     \caption{Interestingness scores of discovered implicit orders.}
     \label{fig:interestingness-scores}
      \paddingD
\end{figure}

\blue{We evaluate our measure of interestingness by reporting the precision of the top-$k$ discovered orders. We define precision as the percentage of correct orders from $m$ randomly selected pairs of values (e.g., for $\A{monthLun}$ if ``Corner $\prec$ Peach'' is correct), verified manually by Computer Science graduate students well-versed in databases. We set $k$ to 5 and $m$ to 10. The precision of our algorithm is 87\% and 89\% over the $\sf{Flight}$ and $\sf{Voter}$ datasets, respectively. 
This experiment demonstrates that the proposed scoring function successfully identifies interesting orders.}

\refstepcounter{exp-c}\label{exp-compute-score-time}
\textbf{Exp-6: Time to compute interestingness.} 
We consider the effect of computing the pairwise measure of interestingness on the algorithm runtime. We observed that in Exp-\ref{exp-scal-tup} in the number of tuples, the runtime increase is less than 1\%.
In Exp-\ref{exp-scal-attr} in the number of attributes, this increase is around 10\% on average, which is attributed to the higher ratio of unique values to the number of tuples.

\subsection{Applications}
\label{sec:exp:apps}


\begin{table}[t]
\begin{small}
    \centering
    \caption{\label{tab:num-of-dependencies} The number of \textbf{implicit orders} (data dependencies).}
    \paddingT
    \scalebox{0.87}{
\begin{tabular}{|l||c|c|c||c|c|c|}
\hline
\multicolumn{1}{|l||}{dataset} & \multicolumn{3}{c||}{Flight} & \multicolumn{3}{c|}{Voter} \\ \hline
\multicolumn{1}{|l||}{type / \# of tuples} & 200K & 600K & 1M & 200K & 600K & 1M \\ \hline \hline
$\FD$ & \textbf{0} (20) & \textbf{0} (20) & \textbf{0} (20) & \textbf{0} (7) & \textbf{0} (5) & \textbf{0} (5) \\ \hline
$\EEOC$ & \textbf{0} (6) & \textbf{0} (5) & \textbf{0} (12) & \textbf{0} (1) & \textbf{0} (1) & \textbf{0} (1) \\ \hline
cond. $\EIOC$ & \textbf{41} (41) & \textbf{31} (31) & \textbf{30} (30) & \textbf{56} (56) & \textbf{23} (23) & \textbf{10} (10) \\ \hline
uncond. $\EIOC$ & \textbf{6} (6) & \textbf{6} (6) & \textbf{6} (6) & \textbf{6} (6) & \textbf{6} (6) & \textbf{6} (6) \\ \hline
cond. $\EIOD$ & \textbf{4} (4) & \textbf{4} (4) & \textbf{4} (4) & \textbf{1} (1) & \textbf{1} (1) & \textbf{1} (1) \\ \hline
uncond. $\EIOD$ & \textbf{3} (3) & \textbf{3} (3) & \textbf{3} (3) & \textbf{1} (1) & \textbf{1} (1) & \textbf{1} (1) \\ \hline
cond. $\IIOC$ & \textbf{6} (3) & \textbf{0} (0) & \textbf{0} (0) & \textbf{4} (2) & \textbf{2} (1) & \textbf{6} (3) \\ \hline
uncond. $\IIOC$ & \textbf{16} (8) & \textbf{16} (8) & \textbf{16} (8) & \textbf{8} (4) & \textbf{2} (1) & \textbf{2} (1) \\ \hline
\end{tabular}
}
\end{small}
\end{table}

\refstepcounter{exp-c}\label{exp-data-profiling}
\blue{
\textbf{Exp-7: Data Profiling.}
}
Table~\ref{tab:num-of-dependencies} illustrates that $\iOrder$ finds implicit $\OD$s in
both $\sf{Flight}$ and $\sf{Voter}$ with 10 attributes up to 1M tuples. 
The second number in the round brackets represents the number of various types of discovered $\OD$s. 
\blue{
While the number of found orders is similar between the datasets with 200K tuples, 
when using 1M tuples, it drops in the $\sf{Voter}$ dataset 
while not changing much in the $\sf{Flight}$ dataset. 
(The increase in the number of conditional $\IIOC$s in the $\sf{Voter}$ dataset is due to having fewer valid candidates at the upper levels of the lattice, and as a result, fewer pruning opportunities for candidates with larger contexts.)
This is consistent with our intuition since the $\sf{Flight}$ dataset 
has more attributes with 
explicit or implicit orders, 
and confirms the importance of a scalable approach to 
invalidate spurious orders as the number of tuples grows.
To investigate the importance of different types of implicit $\OC$s, 
we categorize the number of top-k orders found in the $\sf{Voter}$ dataset 
using each type of implicit $\OC$ in Table~\ref{tab:top-k-types}. 
It can be seen that all types of implicit $\OC$s 
contribute to the most interesting orders found. 
A similar pattern was observed in the $\sf{Flight}$ dataset.
}

\begin{table}
\blue{
\begin{small}
    \centering
     \paddingT
    \caption{\label{tab:top-k-types} \blue{Types of top-k orders.}}
    \paddingT
    \scalebox{0.87}{
    \begin{tabular}{|l|c|c|c|c|c|}
\hline
\textbf{dataset} & \textbf{k} & 
\textbf{\begin{tabular}{@{}c@{}}cond. \\ $\EIOC$/$\OD$ \end{tabular}} & \textbf{\begin{tabular}{@{}c@{}}ucond. \\ $\EIOC$/$\OD$ \end{tabular}} & \textbf{\begin{tabular}{@{}c@{}}cond. \\ $\IIOC$\end{tabular}} & \textbf{\begin{tabular}{@{}c@{}}ucond. \\ $\IIOC$\end{tabular}} \\ \hline
\multirow{3}{*}{\textbf{$\sf{Voter}$}} & \textbf{5} & 0 & 4/1 & 0 & 0 \\ \cline{2-6} 
 & \textbf{10} & 2/1 & 4/1 & 1 & 1 \\ \cline{2-6} 
 & \textbf{15} & 3/1 & 6/1 & 2 & 2 \\ \hline
\end{tabular}
}
 \paddingD
\end{small}
}
\end{table}

\refstepcounter{exp-c}\label{exp-yago}
\textbf{Exp-8: Knowledge Base \blue{Enhancement}.} 
\blue{As another application of implicit domain orders in data profiling,} 
we now compare with an open-source manually curated knowledge-base: YAGO. We quantify the percentage of automatically discovered implicit domain orders by our algorithm among the top-5  (ranked by our pairwise measure of interestingness) that exist in YAGO. 
\blue{The existence of an implicit order in YAGO is evaluated by considering pairs of values within the ordered domain and verifying 
if there exist knowledge triples specifying the relationship between the two entities.
}
The result is that only 20\% of the top discovered orders exist in YAGO. 
\blue{As shown in Exp-\ref{exp-interestingness}, the top orders indicated by our measure of interestingness tend to be meaningful.}
Thus, existing knowledge bases can be enhanced by our techniques, 
\blue{especially in instances where the discovered orders are domain-specific or in knowledge bases that focus on objects rather than concepts, where implicit orders are more common.
This may be done by incorporating pairs of ordered values as knowledge triples, e.g., $(Corner, $ $Less\ than, $ $Peach)$.}

\refstepcounter{exp-c}\label{exp-query-optimization}
\begin{change}
\textbf{Exp-9: Query Optimization.} 
We show that ODs derived from implicit orders can eliminate sort operations in query plans for queries issued by IBM customers.
\footnote{\blue{We thank IBM for providing us with access to these queries.}}  We use the 10GB TPC-DS benchmark in 
Db2\textregistered. 

Q1 (Figure~\ref{query:year-str}) employs a year abbreviation by removing the first two digits from the year and appending 
an apostrophe (e.g., 2002 becomes '02), denoted by the attribute $\A{dYearAbbr}$. 
Our work on implicit orders can be combined with 
the notion of \emph{near-sortedness} \cite{KFM11}. Since the OD 
$\od{\A{dDate}}{\A{dYearAbbr}^*}$ 
is found, the optimizer can then take advantage of
the index on $\A{dDate}$, speeding up the sort operator in the plan 
to accomplish the order-by. Given that each $\A{dYearAbbr}$
block fits into memory, it can be re-sorted in main memory
on-the-fly by $\A{smType}$ and $\A{wsWebName}$.


Next, consider Q2 (Figure~\ref{query:zip-code}). Let there be an index on $\A{sZip}$ in the table $\A{store}$. The first five digits indicate the post office and the following four digits indicate the delivery route. The attribute $\A{sPostOffice}$ denotes the string code of the post office. Since there is the OD $\od{\A{sZip}}{\A{sPostOffice}^*}$, the optimizer can choose to scan the index on $\A{sZip}$ to accomplish the group-by on-the-fly, so no partitioning or sorting is required. Note that a clever programmer cannot rewrite the query manually with a group-by $\A{sZip}$ since $\A{sPostOffice}$ changes the partitioning of the group-by.

\begin{figure}[t]
\paddingT
    \center
    \fontsize{8}{8}\selectfont\ttfamily
\begin{BVerbatim}[commandchars=\\\{\},commentchar=!]
select dYearAbbr, smType, wsWebName,
    sum(case when (ws_ship_date_sk
            - ws_sold_date_sk <= 30)
        then 1 else 0 end) as "30 days", ...
    sum(case when (ws_ship_date_sk
            - ws_sold_date_sk > 120) 
        then 1 else 0 end) as "> 120 days"
from webSales, warehouse, shipMode,
    webSite, dateDim
where wsShipDateSk = dDateSk and ...
group by dYearAbbr, smType, wsWebName
order by dYearAbbr, smType, wsWebName;
\end{BVerbatim}
    \paddingT
    \caption{\col{\blue{Year abbreviation variation with order-by (Q1).}}}
    \label{query:year-str}
    \paddingD
\end{figure}


\begin{figure}[t]
    \center
    \fontsize{8}{8}\selectfont\ttfamily
\begin{BVerbatim}[commandchars=\\\{\},commentchar=!]
select sPostOffice,
    count(distinct sZip) as cnt,
    sum(ssNetProfit) as net
from storeSales, store
where ssStoreSk = sStoreSk
group by sPostOffice;
\end{BVerbatim}
    \paddingT
    \caption{\col{\blue{Post office with group-by (Q2).}}}
    \paddingD
    \label{query:zip-code}
\end{figure}

In Q3 (Figure~\ref{query:to-char-date}), the data are ordered by the attribute $\A{dDateHour}$, which indicates the date in the format ``YYYYDDMM'',  concatenated with a time constant ``12:00:00'' which is commonly used in business intelligence reporting. The optimizer can leverage the OD $\od{\A{dDate}}{\A{dDateHour}^*}$ to use the index on $\A{dDate}$.


\begin{figure}[t]
    \center
    \fontsize{8}{8}\selectfont\ttfamily
\begin{BVerbatim}[commandchars=\\\{\},commentchar=!]
select iItemDesc, dDateHour as when,
    sum(wsSalesPrice) as total
from webSales, item, dateDim
where wsItemSk = iItemSk and
    iCategory = 'Children' and
    wsSoldDateSk = dDateSk
group by iItemDesc, dDateHour
order by dDateHour;
\end{BVerbatim}
    \paddingT
    \caption{\col{\blue{Using string conversion and concatenation (Q3).}}}
    \paddingD
    \label{query:to-char-date}
\end{figure}

Q4 (Figure~\ref{query:year-range}) can have its where-clause rewritten as: $\A{dDate}$ between date('1998-01-01') and date('2002-12-31').  This rewrite allo\-ws the use of an 
index on $\A{dDate}$. Since the OD $\od{\A{dDate}}{\A{dYearAbbr}^*}$ holds (with implicit order `98 $\prec$ `99 $\prec$ $\dots$ $\prec$ `02), the rewritten query is semantically equivalent. 


\begin{figure}[t]
    \center
    \fontsize{8}{8}\selectfont\ttfamily
\begin{BVerbatim}[commandchars=\\\{\},commentchar=!]
select iItemDesc, iCategory,
    iClass, iCurrentPrice,
    sum(wsExtSalesPrice) as revenue
from webSales, item, dateDim
where wsItemSk = iItemSk and
    wsSoldDateSk = dDateSk and
    dYearAbbr in ('98, ..., '02),
group by iItemId, iItemDesc,
    iCategory, iClass, iCurrentPrice
order by iCategory, iClass, iItemId;
\end{BVerbatim}
    \paddingT
    \caption{\col{\blue{With a predicate when abbreviating year (Q4).}}}
    \paddingD
    \label{query:year-range}
\end{figure}

Q5 (Figure~\ref{query:warehouse}) is similar to Q2, but with a filter predicate in its where-clause. The attribute $\A{wWarehouseCode}$ is derived by compr\-essing the first ten characters in the attribute $\A{wWarehouseName}$. 
Using an index on $\A{wWarehouseName}$ and the $\OD$ $\od{\A{wWarehouseName}}{\A{wWarehouseCode}^*}$, we can remove the sort operator. 


\begin{figure}[t]
\paddingT
    \center
    \fontsize{8}{8}\selectfont\ttfamily
\begin{BVerbatim}[commandchars=\\\{\},commentchar=!]
select wWarehouseCode
from webSales, warehouse
where wsWarehouseSk = wWarehouseSk
    and wsQuantity > 90
order by wWarehouseCode;
\end{BVerbatim}
    \paddingT
    \caption{\col{\blue{Warehouse code with order-by (Q5)}}}
    \paddingD
    \label{query:warehouse}
\end{figure}


\begin{figure}[t]
\paddingT
    \center
    \fontsize{8}{8}\selectfont\ttfamily
\begin{BVerbatim}[commandchars=\\\{\},commentchar=!]
select count(*) as count
    over (partition by wsAbsoluteMonth)
from webSales;
\end{BVerbatim}
    \paddingT
\caption{\col{\blue{OLAP query (Q6).}}}
    \paddingD
    \label{query:date-formula}
\end{figure}

Finally, Q6 (Figure~\ref{query:date-formula}) is an OLAP query that uses the partition-by clause over the attribute $\A{wsAbsoluteMonth}$, which is computed from the last two digits of the year multiplied by 100 and adding the value of the month in the numeric format. The query plan can employ the index on $\A{wsSoldDate}$ if the optimizer detects the $\OD$ $\od{\A{wsSoldDate}}{\A{wsAbsoluteMonth}^*}$.


Table~\ref{tab:query-runtimes} shows the execution times in two modes: 
with and without implicit orders rewrites. 
The results for implicit order optimized queries are significantly better,  
with an average decrease in runtime of 30\%.
Our techniques optimize expensive operations such as sort, 
which are super-linear, and which dominate the execution costs as the dataset size increases. 


\balance

\refstepcounter{exp-c}\label{exp-data-mining}
\textbf{Exp-10: Data Mining.}
We next evaluate the impact of implicit orders on the task of data summarization in data mining. 
We use the techniques described in \cite{info-summarization-2014} and \cite{info-summarization-ordinal},
which aim to summarize a database in $k$ rows, each representing a subset of the data, by maximizing the  \emph{information gain} of the summary. 
We consider two scenarios: 
one with all the attributes treated as categorical, and one with implicitly ordered attributes treated as ordinal, allowing the summary to refer to value ranges of these attributes rather than individual values.
We compare the information content of summaries of two datasets, $\sf{Wine}$ and $\sf{Electricity}$, 
which are available in the 
UCI (\url{https://archive.ics.uci.edu/})
and 
openML (\url{https://www.openml.org/})
repositories.
These datasets contain implicitly ordered attributes 
such as $\A{dayOfWeek}$. 
To create more implicit orders for this experiment, 
we converted some numeric attributes into implicitly-ordered ordinal ones 
(e.g., representing price as $cheap$ $\prec$ $moderate$ $\prec$ $expensive$).

There are two parameters in the summarization approach from \cite{info-summarization-ordinal}: a sample size from which to mine summary patterns, and the number of top patterns to explore in each iteration of summary construction. 
We use default values of these parameters: 
4 and 16, respectively, and experiment with values of $k$ (summary size) in $\{1, 2, 4, 8, 16\}$. 
As reported in Figure~\ref{figure:information-gain}, utilizing implicit orders 
allows on average 60\% more information gain using summaries of the same size 
(65\% and 54\% for $\sf{Wine}$ and $\sf{Electricity}$ datasets, respectively). 
The larger improvement on the $\sf{Wine}$ dataset 
could be 
due to 
having more 
attributes with implicit orders in this dataset.


\begin{figure}[t]
    \centering
    \begin{subfigure}[]{0.23\textwidth} 
        \centering
        \includegraphics[width=4.05cm]{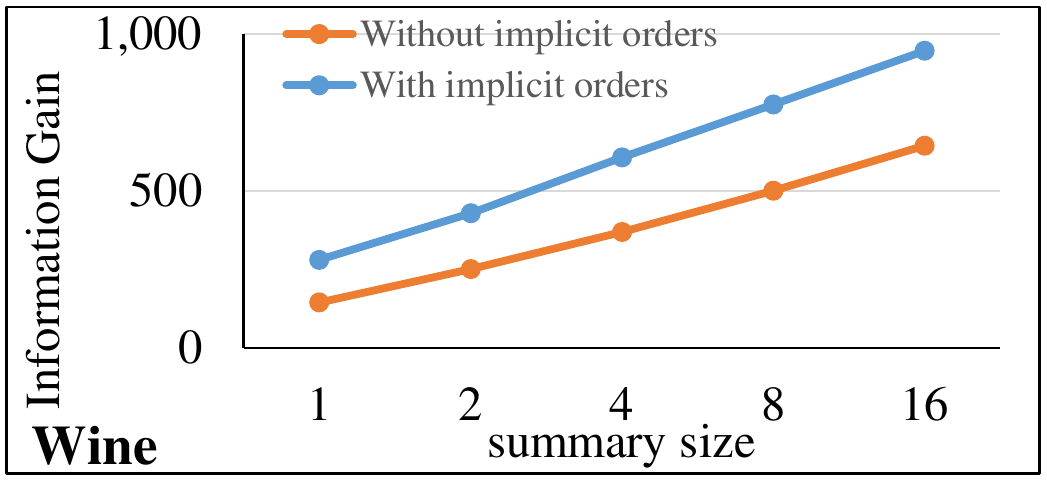}
    \end{subfigure}%
    ~ 
    \begin{subfigure}[]{0.23\textwidth} 
        \centering
        \includegraphics[width=4.05cm]{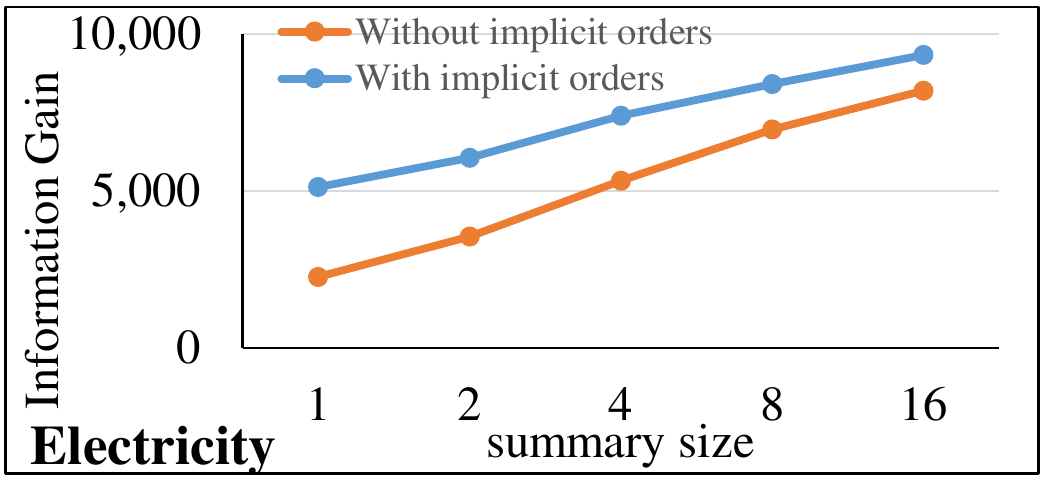}
    \end{subfigure}
    \paddingT
    \caption{\blue{Information gain with and without implicit orders.}}
    \paddingD
    \label{figure:information-gain}
\end{figure}

\end{change}

\begin{table}
\blue{
\begin{small}
    \centering
    \caption{\label{tab:query-runtimes} \blue{Runtimes (seconds) for original and rewritten queries.}}
    \paddingD
    \scalebox{0.87}{
\begin{tabular}{|c|c|c|c|c|c|c|}
\hline
\textbf{Query \#} & \textbf{Q1} & \textbf{Q2} & \textbf{Q3} & \textbf{Q4} & \textbf{Q5} & \textbf{Q6}  \\ \hline
\textbf{Original} & 6.35 & 13.92 & 23 & 16 & 54 & 3.34 \\ \hline
\textbf{Rewritten} & 3.33 & 10.88 & 17 & 14 & 34 & 2.44 \\ \hline
\end{tabular}
}
 \paddingD
\end{small}
}
\end{table}

\section{Related Work} \label{sec:relWork}

Previous work investigated the  properties of and relationships between sorted sets~\cite{DP02}. However, to the best of our knowledge, no algorithms for discovery of implicit domain orders were proposed. 

%

Existing $\OD$ discovery algorithms require some notion of explicit order~\cite{SGG2018:BiODD,SGG+2017:ODD,LN16,JZY20} and can benefit from implicit orders to find ``hidden'' $\OD$s that have not been feasible before.  In our solution, we use the set-based $\OD$ discovery algorithm~\cite{SGG+2017:ODD,SGG2018:BiODD} since other approaches cannot discover a complete set of non-trivial $\OD$s.
For example, the list-based approach in~\cite{LN16} is intentionally incomplete in order to prune the much larger list-based search space.  A similar approach, recently shown in~\cite{JZY20} is also incomplete despite the authors' claim of completeness:   
it omits $\OD$s in which the same attributes are repeated in the left- and the right-hand side, such as $\od{\Lst{ \A{country}, \A{profit}} }{ \Lst{ \A{country}, \A{tax}} }$ and reports an $\OD$ only when both the corresponding $\OFD$ and $\OC$ hold. Thus, it leaves out cases when only an $\OFD$ or only an $\OC$ is true 
(e.g, $\OC$
$\simular{\A{week}}{\A{month}}$
holds, 
but $\OFD$s 
$\ordersC{\A{week}}{\emptyLst}{\A{month}}$ and 
$\ordersC{\A{month}}{\emptyLst}{\A{week}}$ do not hold
over the tuples within a \emph{single} \A{year}). 
Additionally, the algorithm recently presented in~\cite{CSM19} is incomplete, as shown in~\cite{GGK20}.  



The importance of sorted sets has been recognized for query optimization and data cleaning.
In~\cite{QO-exp-orders-vldb05}, the authors explored 
sorted sets for executing nested queries. Sorted sets created as generated columns (SQL functions and algebraic expressions) were used in predicates for query optimization~\cite{Ma2005,Szli2012}. Relationships between sorted attributes 
have been also used to eliminate joins~\cite{Sz2011} 
and to generate interesting orders~\cite{Si1996,SGG+2013:complexity}. A practical application of sorted sets 
to reduce the indexing space was presented in~\cite{Do1982}. 
In~\cite{QTY18}, the authors proved that finding minimal-cost repairs with respect to $\OD$s is NP-complete, 
and introduce an approach to greedy repair.

\section{Conclusions}  \label{sec:futureWork}

We devised the first techniques to discover implicit domain orders. In future work, we plan to study the discovery of \emph{approximate} implicit orders through $\OD$s that hold with some exceptions. 
While in this work, we discover implicit domain orders with respect to a single 
set-based OD, we plan to extend our framework to merge orders found 
for a given attribute with multiple set-based ODs.
\blue{We will also address implicit order discovery in dynamic tables, as was recently done for explicit $\OD$ discovery   \cite{pointwise-od-discovery-vldb20}.}

\newpage


\bibliographystyle{abbrv}
\bibliography{bibliography}  

\newpage
\section{Appendix}\label{sec:appendix}
\subsection{Pseudocode}
Algorithm~\ref{pc:eioc-fd-left-to-right} validates an $\EIOD$ candidate 
$\simular{\A{A}}{\A{B}^*}$ when the corresponding $\FD$ 
$\fd{\A{A}}{\A{B}}$ holds (as in the second case described in Section~\ref{sec:eioc-with-fd}). 
In Lines \ref{pcline:eiod-comp-min} and \ref{pcline:eiod-comp-max}, 
the smallest and largest values of $\A{A}$ co-occurring with each $b_i\in\A{B}$ are found; 
this can be done in a single iteration over the database using a hash table.
In Line~\ref{pcline:eiod-sort}, the values of $\A{B}$ are sorted 
according to the smallest value of $\A{A}$ with which 
they co-occur, and placed in the array $B$.
Lines~\ref{pcline:eiod-begin-loop} to \ref{pcline:eiod-end-loop} check consecutive values in $B$. If the ranges of co-occurring 
$\A{A}$ values for any two consecutive values in $B$ overlap, this $\EIOD$ candidate
is invalid and the algorithm returns ``INVALID''.
Otherwise, the order of values in $B$ is \emph{the} valid implicit total order $\A{B}^*$ 
and is returned.

\begin{algorithm}[ht]
\begin{small}
\caption{EIOD-FD-LeftToRight} 
\label{pc:eioc-fd-left-to-right}
\noindent\textbf{Input:} Attributes $\A{A}$ and $\A{B}$. \hspace{5cm} \ \ \ \ \ \ \ \ \\
\noindent\textbf{Output:} A total order over the values of $\A{B}$ or ``INVALID''. \ \ \ \ \ \ \ \ \ \ \ \ \ \ \ \ \ \ \ \ \ \ \ \ \ \ \ \phantom{} \vspace{-0.3cm}
    \begin{algorithmic}[1]
        \STATE set $\text{\emph{min}}_{\A{A}}(b_i)$ as the smallest value of $\A{A}$ co-occurring with $b_i$ \label{pcline:eiod-comp-min}
        \STATE set $\text{\emph{max}}_{\A{A}}(b_i)$ as the largest value of $\A{A}$ co-occurring with $b_i$ \label{pcline:eiod-comp-max}
        \STATE sort the values of $b_i\in\A{B}$ in ascending order based on $\text{\emph{min}}_{\A{A}}(b_i)$ and store in $B$ \label{pcline:eiod-sort}
        \FOR{each $i$ in $\{1, \dots, \text{\emph{size}}(B)\}$} \label{pcline:eiod-begin-loop}
            \IF{$\text{\emph{min}}_{\A{A}}(B[i + 1])$ $<$ $\text{\emph{max}}_{\A{A}}(B[i])$}
                    \RETURN \emph{``INVALID''}
            \ENDIF
        \ENDFOR \label{pcline:eiod-end-loop}
        \RETURN $B$
    \end{algorithmic}
 \end{small}
\end{algorithm}
\vskip -0.2cm

Algorithm~\ref{pc:deriveChains} derives total orders 
over the values of $\A{A}$ and $\A{B}$ based on the disconnected subtrees 
in $\BGp_{\A{A},\A{B}}$, as described in Section~\ref{sec:iioc-empty-context}. 
Line~\ref{pcline:iioc-empty-context-degree1-node} 
sets \emph{currentNode} to a node with degree one. 
Note that since $\BGp_{\A{A},\A{B}}$ is acyclic, such a node must always exist 
in any subtree $G$ of it.
Lines~\ref{pcline:iioc-empty-context-begin-loop} to \ref{pcline:eiod-end-loop} 
traverse the tree back and forth between the nodes of 
$\A{A}$ and $\A{B}$.
The zigzag chain is
uniquely determined since $\BGp_{\A{A},\A{B}}$ 
has no 3-fan-out, and thus, the degree of each node is at most two.
Depending on the side ($\A{A}$ or $\A{B}$), each node is appended to one of the total orders in Lines~\ref{pcline:iioc-empty-context-append1} or \ref{pcline:iioc-empty-context-append2}, and these two orders are returned as the output of the algorithm.
Using the original graph $\BG_{\A{A},\A{B}}$, the \emph{singleton} nodes can easily be inserted into the derived total orders, resulting in the final \emph{weak} total orders, $\A{A}^*$ and $\A{B}^*$.

\begin{algorithm}[ht]
 \begin{small}
    \caption{IIOC-deriveChains} \label{pc:deriveChains}
  \textbf{Input:} a connected subtree $G$ in $\BGp_{\A{A}, \A{B}}$.   \hspace{3.7cm} \ \ \ \\
  \textbf{Output:} total orders over the values of $\A{A}$ and $\A{B}$ in $G$. \ \ \ \ \ \ \ \ \ \ \ \ \ \ \ \ \ \ \ \ \ \ \ \ \ \   \hphantom{1.7cm} \vspace{-3mm}
    \begin{algorithmic}[1]
        \STATE set \emph{currentNode} to an arbitrary node with degree 1 \label{pcline:iioc-empty-context-degree1-node}
        \STATE set $A$ and $B$ to empty lists
        \STATE set \emph{visited} to the set $\{$\emph{currentNode}$\}$ 
        \STATE append \emph{currentNode} to the appropriate list $A$ or $B$ depending on which side it belongs to \label{pcline:iioc-empty-context-append1}
        \WHILE{\emph{currentNode} has an unvisited neighbor \emph{neighbor}} \label{pcline:iioc-empty-context-begin-loop}
            \STATE add \emph{neighbor} to \emph{visited}
            \STATE append \emph{neighbor} to the appropriate list $A$ or $B$ \label{pcline:iioc-empty-context-append2}
            \STATE set \emph{currentNode} to \emph{neighbor}
        \ENDWHILE \label{pcline:iioc-empty-context-end-loop}
        \RETURN $A, B$
    \end{algorithmic}
 \end{small}
\end{algorithm}

Algorithms~\ref{pc:iioc-strongest-order-first} and \ref{pc:iioc-strongest-order-second} 
demonstrate the main steps to remove edges from the partial order output of 
the SAT solver in order to achieve a pair of strongest derivable orders, 
as described in Section~\ref{SEC:SAT-Solver}. 
Each of these algorithms generates a set of pairs of values (i.e., \emph{nonremovables}), 
which must be kept in the final \emph{strongest} derivable order graph. 
Note that these algorithms describe the steps taken only for values on 
one side of the candidate, which can be repeated for the other side as well.

Algorithm~\ref{pc:iioc-strongest-order-first} corresponds to the first condition described in Section~\ref{SEC:SAT-Solver}.
Lines~\ref{pcline:iioc-nonempty-begin-CC-for} to \ref{pcline:iioc-nonempty-end-CC-for} 
find pairs of values that satisfy this condition, i.e., pairs which exist 
within the same connected component of some $\PG_i$ such that there are at least 
two nodes with degree two or larger along their connecting path.
This is done by running a DFS from every node $l_u$ in the graph and 
detecting all nodes $l_v$ which satisfy this condition.
Algorithm~\ref{pc:iioc-strongest-order-second} corresponds to the 
second condition described in Section~\ref{SEC:SAT-Solver},
i.e., all pairs in a connected component of some $\PG_i$, such that 
there exist two nodes $v_r$ and $v_s$ in 
the same connected component in $\PG_i$ which 
also exist in some other bipartite graph $\PG_j$.
Line~\ref{pcline:iioc-nonempty-samePG} creates a hashmap which, for each 
two values $l_u$ and $l_v$, stores the bipartite graphs in 
which $l_u$ and $l_v$ are connected.
Lines~\ref{pcline:iioc-nonempty-begin-PG-for} to 
\ref{pcline:iioc-nonempty-end-PG-for} store the co-occurrences of 
such pairs of values in the variable \emph{samePG}.
Lines~\ref{pcline:iioc-nonempty-begin-nonrem-for-two} to
\ref{pcline:iioc-nonempty-end-nonrem-for-two} traverse over 
all bipartite graphs to 
detect $\BG_i$s and values $l_p$ and $l_q$
such that $l_p$ and $l_q$ are connected in $\BG_i$ as well as in at least
one other bipartite graph.
Next, all pairs of connected values within such $\BG_i$
are added to the set \emph{nonremovables}.

The remaining step is to compute the union of output sets of
Algorithms~\ref{pc:iioc-strongest-order-first} and 
\ref{pc:iioc-strongest-order-second} and to
traverse over all the edges of the initial partial order graph 
(i.e., the output of the SAT solver) 
and remove those which do not exist in the union set, 
creating the final graph for a strongest derivable order.

\begin{algorithm}[ht]
 \begin{small}
    \caption{IIOC-SameConnectedComponent} \label{pc:iioc-strongest-order-first}
 \textbf{Input:} set of bipartite graphs $\BG$ $=$ $\{\BG_1$, $\dots$, $\BG_k\}$. \hspace{2.2cm} \ \\
 \textbf{Output:} set of pairs of values in one side of the candidate that must be kept in the final order \  \hphantom{1.4cm} 
    \begin{algorithmic}[1]
        \STATE set \emph{nonremovables} as an empty set \label{pcline:iioc-nonempty-nonrem1}
        \FOR{each $\BG_i$ $\in$ $\BG$} \label{pcline:iioc-nonempty-begin-CC-for}
            \FOR{each $l$ $\in$ $\BG_i$}
                \STATE run DFS from $l_v$ and add nodes with degree $\ge2$ to \emph{degTwo} 
                \STATE when visiting node $l_u$:
                \IF{\emph{size}$($\emph{degTwo}$)$ $\ge$ $2$}
                    \STATE add the pair $(l_u,l_v)$ to \emph{nonremovables}
                \ENDIF
            \ENDFOR
        \ENDFOR \label{pcline:iioc-nonempty-end-CC-for}
        \RETURN \emph{nonremovables}
    \end{algorithmic}
 \end{small}
\end{algorithm}

\begin{algorithm}[ht]
 \begin{small}
    \caption{IIOC-MutualPairsOfValuesInPGs} \label{pc:iioc-strongest-order-second}
 \textbf{Input:} set of bipartite graphs $\BG$ $=$ $\{\BG_1$, $\dots$, $\BG_k\}$.  \hspace{2.2cm} \ \\ 
 \textbf{Output:} set of pairs of values in one side of the candidate that must be kept in the final order \  \hphantom{1.4cm} 
    \begin{algorithmic}[1]
        \STATE set \emph{nonremovables} as an empty set \label{pcline:iioc-nonempty-nonrem2}
        \STATE set \emph{samePG} as a hashmap from pairs of values in $L$ to empty sets \label{pcline:iioc-nonempty-samePG}
        \FOR{each $\BG_i$ $\in$ $\BG$} \label{pcline:iioc-nonempty-begin-PG-for}
            \FOR{each $l_u$, $l_v$ $\in$ $\BG_i$}
                \STATE add $i$ to $\text{\emph{samePG}}[l_u,l_v]$
            \ENDFOR
        \ENDFOR \label{pcline:iioc-nonempty-end-PG-for}
        \FOR{each $\BG_i$ $\in$ $\BG$} \label{pcline:iioc-nonempty-begin-nonrem-for-two}
            \IF{there exist $l_p$, $l_q$ $\in$ $\BG_i$ such that $\text{\emph{samePG}}[l_p,l_q]$ $\ge$ 2}
                \FOR{each $l_u$, $l_v$ $\in$ $\BG_i$}
                    \STATE add the pair $(l_u,l_v)$ to \emph{nonremovables}
                \ENDFOR
            \ENDIF
        \ENDFOR \label{pcline:iioc-nonempty-end-nonrem-for-two}
        \RETURN \emph{nonremovables}
    \end{algorithmic}
 \end{small}
\end{algorithm}

\subsection{Proofs}

\setcounter{lemma}{0}
\setcounter{theorem}{0}

\begin{theorem}
Assume the \FD\ of \fd{\A{A}}{\A{B}}\ holds in \T{r}.
Then
\(\simular{\A{A}}{\A{B}^{*}}\) holds
\emph{iff}
\((\pi_{\A{B}})_{\A{A}}\)
is an interval partitioning;
\Ord{T}[{\A{B}_{*}}][\prec]\
is the unique order of \A{B}'s values
corresponding to their order
in \T{r}\ as sorted by \A{A}.
\end{theorem} 

\begin{proof}
Assume $(\pi_{\A{B}})_{\A{A}}$ is an interval partitioning and, without loss of generality, let $a_1$, $\dots$, $a_{p_q}$ be the explicit order over values in $\A{A}$ and the value of tuples $\tup{t}_1$ to $\tup{t}_{p_q}$ projected on columns $\A{A}$ and $\A{B}$ be the following: 
$(a_1, b_1)$, $(a_2, b_1)$, $\dots$, $(a_{p_1}, b_1)$, $(a_{p_1+1}, b_2)$, $\dots$, $(a_{p_2}, b_2)$, $\dots$, $(a_{p_{q-1}+1}, b_q)$, $\dots$, $(a_{p_q}, b_q)$ (tuples with duplicate values are represented as one tuple for simplicity). 
Therefore, $\set{E}_{j}(\tup{t}_{\A{B}})_{\A{A}}$ $=$ $\{a_{p_{j-1}+1}$, $\dots$, $a_{p_j}\}$ (assuming $p_0=0$). 
The total order $\Ord{T}[\T{r}][\prec]$, $\tup{t}_1$ $\prec$ $\dots$ $\prec$ $\tup{t}_{p_q}$ is a valid \emph{witness} and its projection on $\A{B}$ results in total order $b_1$ $\prec$ $\dots$ $\prec$ $b_q$. 
Since this is the only valid order (as it is enforced by the explicit order over the values of $\A{A}$), \Ord{T}[{\A{B}_{*}}][\prec] is the only valid---and therefore the strongest derivable---implicit order over the values of $\A{B}$.

If $(\pi_{\A{B}})_{\A{A}}$ is not an interval partitioning, then 
there exist \(i, j \in [1,\ldots,k]\)
such that $i<j$ and 
\(
	\mbox{min}(\set{E}_{i}(\tup{t}_{\A{B}})_{\A{A}})
		\prec
	\mbox{min}(\set{E}_{j}(\tup{t}_{\A{B}})_{\A{A}})
		\prec
	\mbox{max}(\set{E}_{i}(\tup{t}_{\A{B}})_{\A{A}})
\).
Let $\tup{t}_{i_1}=(\mbox{min}(\set{E}_{i}(\tup{t}_{\A{B}})_{\A{A}}), b_i)$,
$\tup{t}_{i_2}=(\mbox{max}(\set{E}_{i}(\tup{t}_{\A{B}})_{\A{A}}), b_i)$, and
$\tup{t}_{j}=(\mbox{min}(\set{E}_{j}(\tup{t}_{\A{B}})_{\A{A}}), b_j)$. 
If this $\EIOD$ candidate is valid, then a total order $\Ord{T}[\T{r}][\prec]$ must exist that is: 1) compatible with the explicit order over the values of $\A{A}$ and 2) its projection produces a valid total order over the values of $\A{B}$. 
Given the first condition, in any total order $\Ord{T}[\T{r}][\prec]$, there must be $\tup{t}_{i_1}$ $\prec$ $\tup{t}_j$ $\prec$ $\tup{t}_{i_2}$, while the projection of this order would produce a relation over $\A{B}$ with $b_i\prec b_j$ and $b_j\prec b_i$, which cannot be a valid total order. Therefore, by contradiction, this $\EIOC$ candidate is not valid.
\end{proof}

Let 
the $\FD$ be 
\fd{\A{A}}{\A{B}},
\(m = |\A{B}|\) (the number of distinct values of \A{B}),
and
\(n = |\T{r}|\) (the number of tuples).

\begin{lemma}
The runtime of discovering $\EIOD$s with an empty 
context is ${\mathcal O}(m\ln m + n)$.
\end{lemma}


\begin{proof}
Let the attributes involved in the $\EIOD$ candidate
with an empty context be $\A{A}$ and $\A{B}$, and 
assume the $\FD$ $\fd{\A{A}}{\A{B}}$ holds.
There are three cases (as enumerated in Section~\ref{sec:eioc-with-fd}) that
depend on which side has the implicit order 
and whether the $\FD$ $\fd{\A{B}}{\A{A}}$ holds. 
In all cases, we only need to sort $m$ values 
(the number of distinct values of $\A{B}$), 
for a cost of $m\ln m$.
The other steps can be done in linear time using a hash table. 
Thus, the total runtime is ${\mathcal O}(m\ln m + n)$.
\end{proof}

\begin{theorem}\label{thm:emptyContEI2}
$\simular{\A{A}}{\A{B}[][*]}$ is 
valid
\emph{iff}
$\A{B}[][{\tau_{\A{A}}}]$ is a 
weak total
order.
\end{theorem} 

\begin{proof}
Assume $\A{B}[][{\tau_{\A{A}}}]$ is a valid weak total order 
and let $\Ord{T}[\A{B}^*][\prec]$ denote an arbitrary
strong total order compatible with---that is, is a superset of---$\A{B}[][{\tau_{\A{A}}}]$.
To find a \emph{witness} order $\Ord{T}[\T{r}][\prec]$ 
over the \emph{tuple}s of the database, sort the tuples based on the 
explicit order over $\A{A}$ and break ties by $\Ord{T}[\A{B}^*][\prec]$. 
Considering the projected order of this witness over
$\A{B}$, 
the $b_i\prec b_j$ derived \emph{within} a partition group 
of $\A{A}$ cannot result in a cycle, as they were initially 
enforced by $\Ord{T}[\A{B}^*][\prec]$. 
For the $b_i\prec b_j$ derived from
\emph{different} partition groups of $\A{A}$, it suffices to
verify \emph{consecutive} partition groups of $\A{A}$ 
due to the \emph{transitivity} of the order relation.
Assume the relation $b_i\prec b_j$ can be derived from two tuples in different
partition groups $\PG_p$ and $\PG_q$ where, without loss of generality, $p < q$. 
If $q=p+1$, these partition groups are consecutive and $b_i\prec b_j$
is derived using our algorithm. 
Otherwise, the partition group 
$\PG_{p+1}\ne\PG_q$ contains some tuple with the $\A{B}$-value of $b_k$ (note that $b_k$ could be the same as $b_i$ or $b_j$) since otherwise $\PG_p$ and $\PG_q$ would 
have been consecutive. 
Using our algorithm, the order $b_i\prec b_k$ is derived (if $b_i\ne b_k$).
Due to transitivity, it is now only required for $b_k\prec b_j$ 
to be derived (if $b_j\ne b_k$),
which is inductively derived from $\PG_{p+1}$ and $\PG_{q}$. 
Therefore, the validity of $\A{B}[][{\tau_{\A{A}}}]$ 
implies the validity of the witness order over tuples.

Now we prove the other direction. Without loss of generality, let $\PG_i$ denote 
the $i$-th sorted partition group of $\A{A}$ and $\A{B}_i$ the 
set of values of $\A{B}$ that co-occur with $\PG_i$.
Assume $\simular{\A{A}}{\A{B}^*}$ is valid and that $\Ord{T}_{\T{r}}^{\prec}$ 
is some witness order over the tuples.
Based on the definition of a witness order, the partition groups of 
$\A{A}$ and $\A{B}$ need to be placed consecutively 
(and in ascending order for $\A{A}$)
in $\Ord{T}_{\T{r}}^{\prec}$.
Since the projected order of $\Ord{T}_{\T{r}}^{\prec}$ over $\A{B}$
is a valid total order, for any two
consecutive $\PG_i$ and $\PG_{i+1}$, $\A{B}_i$ and $\A{B}_{i+1}$ 
have at most one value in common. 
For each three consecutive partition groups 
$\PG_{i-1}$, $\PG_i$, and $\PG_{i+1}$, 
let $\A{B}_{i-1,i}=\A{B}_{i-1}\cap\A{B}_i$, 
$\A{B}_{i, i+1}=\A{B}_{i}\cap\A{B}_{i+1}$, and
$\A{B}'_i=\A{B}_i\setminus\A{B}_{i-1}\setminus\A{B}_{i+1}$.
Since $\Ord{T}_{\T{r}}^{\prec}$ is a valid order, for any two
$i\ne j$, $\A{B}'_i\cap\A{B}'_j=\emptyset$.
Therefore, by deriving pairs of ordered
values from consecutive partition groups of $\A{A}$, 
the resulting order graph corresponds to a valid weak total order where each partition 
over the values of $\A{B}$ correspond to the value(s) within a set 
$\A{B}'_i$ or $\A{B}_{i,i+1}$ for some $i$ 
(in cases where $\A{B}'_i=\emptyset$ and $\A{B}_{i-1,i}=\A{B}_{i,i+1}$,
consecutive partitions with the same values can be considered as one).
\end{proof}

In the remainder of this section, $n$, $m$, and $p$ denote
the number of tuples, the number of distinct values of the candidate attribute(s) with 
implicit order, and the number of partition groups of the context, respectively.

\begin{lemma}
The runtime of discovering $\EIOC$s with an empty context is ${\mathcal O}(n + m^2)$, given an initial sorting of the values in the first level of the lattice.
\end{lemma}

\begin{proof}
Consider the $\EIOC$ candidate $\simular{\A{A}}{\A{B}^*}$.
Since the attributes
in the \OD-discovery algorithm 
in Section~\ref{sec:prelims} 
have been sorted in advance
for the first level of the lattice,
the sorted partition groups over the values of $\A{A}$
can be created in ${\mathcal O}(n)$ time. 
A similar argument about the \emph{3-fan-out} rule 
for $\IIOC$s 
in Section~\ref{sec:iioc-empty-context} also applies to $\EIOC$s with an 
empty context as $\EIOC$s are more restrictive than $\IIOC$s. 
This ensures that when traversing consecutive 
partition groups of $\A{A}$ and inferring relations of the form
$b_i\prec b_j$, the number of these relations will be bounded by $m^2$.
Since the graph storing the partial order over the values of $\A{B}$
will at most have size ${\mathcal O}(m^2)$ as well, the total runtime for 
an $\EIOC$ candidate with an empty context is ${\mathcal O}(n+m^2)$.
\end{proof}


\begin{theorem}
There exists an implicit domain order
\(\Ord{P}[{\A{B}[][*]}][\prec]\)
such that the \EIOC\
\(\simularCtxSet{\set{X}}{\A{A}}{\A{B}[][*]}\) holds
\emph{iff}
the union graph is cycle free.
\end{theorem}

\begin{proof}
Let $G_i$ denote the set of strongest derivable orders from within each partition group, 
and let $G$ be the \emph{union graph} generated from these graphs using
the procedure described in Section~\ref{section:nonemptyEI}.
Each $G_i$ corresponds to a valid partial order over the values of the attribute with the implicit order.
First, assume $G$ corresponds to a valid partial order; i.e., it is acyclic. 
The \emph{witness} order over the tuples within each partition group represents a valid witness order over the entire dataset, since the union of all the derived relations over the values corresponds to a valid partial order. 

Assume that $G$ is cyclic and, without loss of generality, let $\{(b_1, b_2)$, $(b_2, b_3)$, $\dots$, $(b_{p-1},b_p)$, $(b_p,b_1)\}$ denote the set of edges involved in a cycle in this graph.
Clearly, each edge $(b_{j}, b_{k})$ in this cycle must exist within some graph $G_i$. 
Since the order graphs $G_i$ correspond to the \emph{strongest} derivable order over each partition group, an edge $(b_j, b_k)$ $\in$ $G_i$ must be present in \emph{all} witness total orders of the corresponding partition group.
Therefore, for any witness order over the partition groups, and consequently over the entire dataset, $b_j \prec b_k$ for any $b_j$ and $b_k$ will be derived, meaning that the partial order derived over the entire dataset cannot be valid.
Thus, the unconditional $\EIOC$ candidate with a non-empty context is invalid.
\end{proof}

\begin{lemma}
The time complexity of discovering $\EIOC$s with a non-empty context is ${\mathcal O}(n\ln n+pm^2)$.
\end{lemma}

\begin{proof}
Given an $\EIOC$ candidate with a non-empty context $\set{X}$, the algorithms 
in Sections~\ref{sec:eioc-with-fd} or \ref{sec:eioc-empty-context} run 
for each partition group over $\set{X}$.
The runtime of discovering an implicit order over each partition group 
is ${\mathcal O}(k\ln k+m^2)$, where $k$ denotes the number of tuples in the partition group.
Furthermore, the size of the order graph constructed from each partition group 
is ${\mathcal O}(m^2)$. Since creating the union graph using the individual orders and 
checking for cycles can 
be done in linear time in size of the graphs, the total runtime is ${\mathcal O}(pk\ln k + pm^2)$ which is bounded by ${\mathcal O}(n\ln n+pm^2)$.
\end{proof}

\begin{theorem} \label{theo:bipValidity2}
\oc{\A{A}[][*]}{\A{B}[][*]}\
is \emph{valid}
over \T{r}\
\emph{iff}
the following two conditions are \emph{true}
for \BGp[\A{A}, \A{B}]\ over \T{r}:

\begin{enumerate}[nolistsep,leftmargin=*]
\item it contains no 3-fan-out;
	and
\item it is acyclic.
\end{enumerate}
\end{theorem}

\begin{proof}
Assume that $\BGp[\A{A}, \A{B}]$ contains a 3-fan-out. Without loss of generality, 
let the tuples participating in this 3-fan-out be as follows 
(the minimal required tuples for a node $a_4$ containing the 3-fan-out): 
$\tup{t}_1=(a_1, b_1)$, $\tup{t}_2=(a_2, b_2)$, $\tup{t}_3=(a_3, b_3)$, 
$\tup{t}_4=(a_4, b_1)$, $\tup{t}_5=(a_4, b_2)$, and $\tup{t}_6=(a_4, b_3)$. 
Assume that a \emph{witness} total order
$\Ord{T}[\T{r}][\prec]$ exists over the tuples, and, without loss of generality, 
$\tup{t}_4$ $\prec$ $\tup{t}_5$ $\prec$ $\tup{t}_6$. Since the projections of
$\Ord{T}[\T{r}][\prec]$ over both $\A{A}$ and $\A{B}$ 
need to result in a valid total order, it must be that $\tup{t}_1$ $\prec$ $\tup{t}_4$ 
and $\tup{t}_6$ $\prec$ $\tup{t}_3$. However, it is  
impossible to find a valid placement for $\tup{t}_2$, as any placement 
results in a cycle in the projected order over $\A{A}$ or $\A{B}$. 
Therefore, by contradiction, the candidate is invalid. 

Assume $\BGp[\A{A}, \A{B}]$ contains a cycle. 
Let a cycle with length $k$ be over tuples $\tup{t}_1=(a_1, b_1)$, 
$\tup{t}_2=(a_2, b_1)$, $\dots$, 
$\tup{t}_{k-1}=(a_{k/2},b_{k/2})$, and $\tup{t}_k=(a_1, b_{k/2})$, 
without loss of generality. 
Assume that a \emph{witness} total order $\Ord{T}[\T{r}][\prec]$ 
exists over the tuples, and without loss of generality, 
$\tup{t}_1$ $\prec$ $\tup{t}_2$. 
Then $\tup{t}_2$ $\prec$ $\tup{t}_3$, and, inductively, 
$\tup{t}_{i}$ $\prec$ $\tup{t}_{i+1}$ must hold. 
Thus, $\tup{t}_{1}$ $\prec$ $\tup{t}_2$ $\prec$ $\dots$ $\prec$ $\tup{t}_k$. 
However, this results in $a_1$ $\prec$ $a_{k/2}$ and $a_{k/2}$ $\prec$ $a_1$, 
which makes a valid order over $\A{A}$ impossible. 
Therefore, by contradiction, the candidate is invalid.

The other direction follows the
correctness argument of Algorithm~\ref{pc:deriveChains}. 
Let $\Ord{T}[\A{A}'^*][\prec]$ and $\Ord{T}[\A{B}'^*][\prec]$ 
denote the orders derived using Algorithm~\ref{pc:deriveChains} 
(which are unique modulo polarity). 
Without loss of generality, let $a_1\prec a_2\prec\dots\prec a_k$ and 
$b_1\prec b_2\prec\dots\prec b_k$ denote these total orders 
(the case where the length of one of the chains is longer by \emph{one}
value is resolved similarly).
Without loss of generality, let $a_1$ denote the first node in 
$\BGp[\A{A}, \A{B}]$ with degree 1. 
This implies that the tuples need to be of either of two patterns:
1) $(a_1, b_1)$; or 2) $(a_i, b_{i-1})$ and $(a_i, b_i)$ for $i>1$.
The total order $\Ord{T}[\T{r}'][\prec]$ 
derived by sorting the tuples by
$\Ord{T}[\A{A}'^*][\prec]$ and breaking ties by 
$\Ord{T}[\A{B}'^*][\prec]$ results in a valid \emph{witness} order.
This is because any cycles in the derived orders over $\A{A}'$ or $\A{B}'$
would indicate the existence of a cycle 
(e.g., some $a_i$ occurring with $b_{i-2}$)
or a 3-fan-out 
(e.g., some $a_i$ occurring with three values 
$b_{i-1}$, $b_i$, and $b_{i+1}$).
For the \emph{singleton} values removed in $\A{A}'$ and $\A{B}'$, 
they can easily be added to $\Ord{T}[\T{r}'][\prec]$, 
where a tuple $(a_i, b_k)$ can be added between the tuples
$(a_i, b_{i-1})$ and $(a_i, b_i)$, resulting in the 
valid final witness order $\Ord{T}[\T{r}][\prec]$, 
implying the validity of this $\IIOC$ candidate.
\end{proof}

\begin{lemma}
The runtime of validating a conditional $\IIOC$ with an empty or a non-empty context is ${\mathcal O}(n)$.
\end{lemma}


\begin{proof}
To validate an \IIOC\ candidate with an empty context
involves generating the \BG,
iterating over the tuples once, and
then using DFS traversal to check
for cycles and 3-fan-outs.  This
can be done in linear time in the number of tuples.
Note that if the 3-fan-out condition holds,
the size of $\BG$ is linear in the number of distinct values ($m$). 
Thus, deriving an order can be 
done in ${\mathcal O}(m)$ time, as it only requires traversing the bipartite graph
a constant number of times,
as shown in Algorithm~\ref{pc:deriveChains}.
For non-empty contexts,
validation of \IIOC s requires
the above steps for each partition group.
Thus,
the overall runtime remains
\({\mathcal O}(n)\).
\end{proof}



\begin{lemma}\label{lemma:NPcomplete2}
The Chain Polarization Problem is NP-Complete.
\end{lemma}

\begin{proof}
The input size of a CPP instance may be measured as the sum of the lengths of its lists; let this be $n$. Consider a pair explicitly implied by the list collection to be in the binary ordering relation if the pair of elements appears immediately adjacent in one of the lists. Thus, the number of explicitly implied pairs is bounded by $n$.

\textbf{CPP is in the class NP.}
An answer of yes to the corresponding decision question means there exists a polarization of the CPP instance that admits a strong partial order. Given such a polarization witness, its validity can be checked in polynomial time. The size of the polarization is at most $n$. The set of explicitly implied ordered-relation pairs is at most $n$. Computing the transitive closure over this set of pairs is then polynomial in $n$. If no reflexive pair (e.g., a $\prec$ a) is discovered, then there are no cycles in the transitive closure, and thus this represents a strong partial order. Otherwise, not.

\textbf{CPP is NP-complete.}
The known NP-complete problem NAE-3SAT (Not-All-Equal 3SAT) can be reduced to CPP.

The structure of a NAE-3SAT instance is a collection of clauses. Each clause consists of three literals. A literal is a propositional variable or the negation thereof. A clause is interpreted as the disjunction of its literals, and the overall instance is interpreted as the logical formula which is the conjunction of its clauses. Since each clause is of fixed size, the size of the NAE-3SAT instance may be measured by its number of clauses; call this $n$.

The decision question for NAE-3SAT is whether there exists a truth assignment to the propositional variables (propositions) that satisfies the instance formula such that, for each clause, at least one of its literals is false in the truth assignment and at least one is true. 

We can establish a mapping from NAE-3SAT instances to CPP instances which is polynomial time to compute and for which the decision questions are synonymous.
Let $p_{1}, .., p_{m}$ be the propositions of the NAE-3SAT instance. For clause \emph{i}, let ($L_{i,1}, L_{i,2}, L_{i,3}$) represent it, where each $L_{i,j}$ is a placeholder representing the corresponding proposition or negated proposition as according to the clause.
We build a corresponding CPP instance as follows. For each clause, we add three lists. For clause \emph{i}, add
[$X_{i,1}, a_i, b_i, Y_{i,1}$],
[$X_{i,2}, b_i, c_i, Y_{i,2}$], and
[$X_{i,3}, c_i, a_i, Y_{i,3}$].
Each $X_{i,j}$ and $Y_{i,j}$ are placeholders above, and correspond to the $L_{i,j}$’s in the clauses.

We replace them in the lists as follows. There are two cases for each: $L_{i,j}$ corresponds to a proposition or the negation thereof. Without loss of generality, let $L_{i,j}$ correspond to $p_k$ or to $\neg p_k$. If $L_{i,j}$ corresponds to $p_k$, we replace $X_{i,j}$ with element $t_k$ and $Y_{i,j}$ with element $f_k$. Else ($L_{i,j}$ corresponds to $\neg p_k$), we replace $X_{i,j}$ with $f_k$ and $Y_{i,j}$ with $t_k$.
Call the $t_i$ and $f_i$ elements propositional elements, and call the $a_i$, $b_i$, and $c_i$ elements confounder elements.

There exists a polarization of the CPP instance that admits a strong partial order iff the corresponding NAE-3SAT instance is satisfiable such that, for each clause, at least one of its literals has been assigned false.
For the propositional elements, let us interpret $t_i \prec f_i$ in the partial order as assigning proposition $p_i$ as true; and $f_i \prec t_i$ as assigning it false.

For any polarization of the CPP instance that admits a strong partial order, for each clause, $i$, one or two, but not all three, of the corresponding lists must have been reversed. Otherwise, there will be a cycle in the ordering relation over the confounder elements, $a_i$, $b_i$, and $c_i$.
Since not all three lists for clause $i$ can be reversed, $X_{i,1} \prec Y_{i,1}, X_{i,2} \prec Y_{i,2}$, or $X_{i,3} \prec Y_{i,3}$. And thus, at least one of the clause’s literals is marked as true, and so the clause is true. Since at least one of the three lists for clause i must be reversed, we know that $Y_{i,1} \prec X_{i,1}, Y_{i,2} \prec X_{i,2}$, or $Y_{i,3} \prec X_{i,3}$. And thus, not all of the clause’s literals are marked as true.
If two lists (coming from different clauses) “contradict” — that is, one implies, say, $t_i \prec f_i$ and the other that $f_i \prec t_i$ — then there would be a cycle in the transitive closure of the ordering relation. This would contradict that our polarization admits a strong partial order. Thus, for each $p_i$, the partial order either has $t_i \prec f_i$ or $f_i \prec t_i$. This corresponds to a truth assignment that satisfies the NAE-3SAT instance.

In the other direction, if no polarization of the CPP instance admits a strong partial order, then there is no truth assignment that satisfies the NAE-3SAT instance such that, for each clause, at least one of its literals has been assigned false.
\end{proof}

\stepcounter{theorem}

\begin{theorem}
If \fd{\set{X}\A{A}}{\A{B}},
then the conditional \IIOD\ candidate
\oc[\set{X}]{\A{A}[][*]}{\A{B}[][*]}\ must be valid.
Furthermore,
there is a unique partial order
that can be derived
for \A{A}[][*]\ and for \A{B}[][*]:
the empty order.
\end{theorem}

\begin{proof}
Since the $\IIOD$ candidate is considered conditionally, 
the validity of the candidates and the \emph{strongest derivable} 
orders over $\A{A}^*$ and $\A{B}^*$ are derived independently from 
within each partition group. 
Consider an arbitrary partition group and, since $\fd{\set{X}\A{A}}{\A{B}}$, 
without loss of generality, let the value of tuples $\tup{t}_1$ to 
$\tup{t}_{p_q}$ in the columns $\A{A}$ and $\A{B}$ be the following: 
$(a_1, b_1)$, $(a_2, b_1)$, $\dots$, $(a_{p_1}, b_1)$, $(a_{p_1+1}, b_2)$, 
$\dots$, $(a_{p_2}, b_2)$, $\dots$, $(a_{p_{q-1}+1}, b_q)$, $\dots$, 
$(a_{p_q}, b_q)$ (tuples with duplicate values are represented as 
one tuple for simplicity). 
The total order $\Ord{T}[\T{r}][\prec]$ $\tup{t}_1$ $\prec$ $\dots$ 
$\prec$ $\tup{t}_{p_q}$ over the tuples is a valid \emph{witness} and 
its projection would result in total orders $a_1$ $\prec$ $\dots$ 
$\prec$ $a_{p_1}$ and $b_1$ $\prec$ $\dots$ $\prec$ $b_q$ over 
the values of $\A{A}$ and $\A{B}$, respectively. The same argument 
can be applied to the rest of the partition groups, meaning that a 
\emph{conditional} $\EIOD$ candidate is always valid.

As for the \emph{strongest derivable} orders over $\A{A}^*$ and $\A{B}^*$, 
without loss of generality, let the initial witness order be 
$\Ord{T}[\T{r}][\prec]$ and consider \emph{transposition}s of 
either of these types applied on it: 1) relocating the tuple 
$\tup{t}_i$, $p_{j-1}+1$ $\le$ $i$ $\le$ $p_{j}$, within the range 
$[p_{j-1}+1$, $p_{j}]$; 2) relocating the consecutive set of tuples 
$\tup{t}_{p_{j-1}+1}$, $\dots$, $\tup{t}_{p_j}$ to any location either 
between $\tup{t}_{p_k}$ and $\tup{t}_{p_k+1}$, or the beginning or end of
$\Ord{T}[\T{r}][\prec]$. Both transpositions are \emph{valid} since the resulting projected orders over 
the values of $\A{A}$ and $\A{B}$ are valid. The intersection 
of all such transposed witness orders for each side is the empty partial order. 
Therefore, the strongest derivable orders over $\A{A}$ and $\A{B}$ 
for this partition group (and, subsequently, for the other partition groups) 
are empty partial orders.
\end{proof}

\begin{theorem}
The unconditional $\IIOC$ candidate is valid \emph{iff} the corresponding SAT instance is satisfiable.
\end{theorem}

\begin{proof}
Assume the SAT instance is satisfiable.
A truth assignment to the SAT instance variables that satisfies 
all the clauses can be translated to two valid 
implicit orders for 
the $\IIOC$ instance, by considering the variables for each pair of values, 
where a \emph{true} assignment to $l_{u,v}$ or $l_{v,u}$ 
(or $r_{u, v}$ and $r_{v,u}$, accordingly) 
would result in the implied order $l_u\prec l_v$ or $l_v\prec l_u$, respectively.
Exactly one of these two variables is set to \emph{true}, given 
the initial clauses generated for each pair of variables.
To derive a witness order $\Ord{T}[\T{r}][\prec]$, it is enough to sort
the tuples in each partition group by the order derived 
over the values of LHS, and break ties by the order
derived over the RHS (or vice versa).
The projection of $\Ord{T}[\T{r}][\prec]$ 
over the LHS and RHS attributes will result in a valid total order,
as all pairs of tuples within each partition group were
considered when generating the \emph{no swap} clauses, meaning that 
the projected order over the attributes is perfectly captured by these clauses.
Furthermore, \emph{transitivity} ensures that the projected order is acyclic, and as a result, is a valid order.

In the other direction, assume that the $\IIOC$ candidate is valid and
let $\Ord{T}[\T{r}][\prec]$ denote a valid witness
order over the tuples. To derive a solution to the SAT instance, 
the reverse of the previous algorithm can be performed; i.e.,
the truth assignment for variables $l_{u,v}$ and $l_{v,u}$ 
(or $r_{u, v}$ and $r_{v,u}$) 
can be set based on the order between values in the projected 
order  
(if no order between the values exists, both variables are set to \emph{false}). 
Since $\Ord{T}[\T{r}][\prec]$ is a valid witness, the projected order over the 
attributes will be a valid order, meaning that the truth assignments
for the SAT variables can be uniquely determined, and would satisfy the
initial conditions (since for each two values $l_u$ and $l_v$, at most 
one of $l_u\prec l_v$ or $l_v\prec l_u$ holds in the projected orders).
Furthermore, this variable assignment would also satisfy the \emph{no swaps}
and \emph{transitivity} clauses, as the projected order over the attributes
is directly derived from \emph{pair}s of tuples and is \emph{acyclic}, respectively.
Therefore, the SAT instance is satisfiable. 
\end{proof}

\begin{lemma}
The cost of the SAT reduction is 
${\mathcal O}(n+pm^2+m^3)$.
\end{lemma}

\begin{proof}
The SAT representation has ${\mathcal O}(m^2)$ propositional variables. We initially 
also generate ${\mathcal O}(m^2)$ clauses, one for each pair of variables, $v_{i,j}$ and $v_{j,i}$. 
Generating the no-swap clauses for each $\sf{BG}$ takes ${\mathcal O}(m^2)$ time,  
as the number of edges in the bipartite graph derived from 
each partition group is ${\mathcal O}(m)$, since the initial $\sf{BG}$ is acyclic and 
does not contain any \emph{3-fan-out}s. This makes the runtime of 
this step (and the number of generated clauses) ${\mathcal O}(pm^2)$, where $p$ denotes 
the number of partition groups. Adding the transitivity 
clauses takes ${\mathcal O}(m^3)$ time. This makes the total cost of the reduction 
to the SAT problem ${\mathcal O}(n+pm^2+m^3)$. 
\end{proof}



\end{document}